\documentclass[12pt]{article}

\usepackage{amsmath,amssymb,amsthm}
\usepackage[margin=1in]{geometry}
\usepackage{mathtools}

\usepackage{multirow}

\newtheorem{proposition}{Proposition}

\newcommand{\beqa}{\begin{eqnarray*}}
\newcommand{\eeqa}{\end{eqnarray*}}
\newcommand{\beqn}{\begin{eqnarray}}
\newcommand{\eeqn}{\end{eqnarray}}
\newcommand{\baa}{\begin{array}}
\newcommand{\eaa}{\end{array}}
\newcommand{\bcc}{\begin{center}}
\newcommand{\ecc}{\end{center}}
\newcommand{\btab}{\begin{tabular}}
\newcommand{\etab}{\end{tabular}}

\def\b1e{{\mathbf e}}
 
\newcommand{\Var}{{\rm Var}}

\newcommand{\bSigma}{\boldsymbol{\Sigma}}

\newcommand{\bbeta}{\boldsymbol{\beta}}

\newcommand{\bK}{\boldsymbol{K}}
\newcommand{\bA}{\boldsymbol{A}}
\newcommand{\bB}{\boldsymbol{B}}

\newcommand{\bba}{\boldsymbol{b}}

\newcommand{\bV}{\boldsymbol{V}}

\newcommand{\bmm}{\boldsymbol{m}}

\newcommand{\bk}{\boldsymbol{k}}

\newcommand{\bx}{\boldsymbol{x}}
\newcommand{\bX}{\boldsymbol{X}}
\newcommand{\by}{\boldsymbol{y}}

\newcommand{\bZ}{\boldsymbol{Z}}
\newcommand{\bz}{\boldsymbol{z}}

\renewcommand{\bba}{\boldsymbol{b}}

\newcommand{\BC}{\mathrm{BC}}
\newcommand{\E}{\mathbb{E}}
\newcommand{\R}{\mathbb{R}}

\newtheorem{assumption}{Assumption}

\newcommand{\Prob}{\mathbb{P}}

\newcommand{\pto}{\xrightarrow{p}}
\newcommand{\dto}{\xrightarrow{d}}

\newcommand{\IF}{\mathrm{IF}}
\newcommand{\ELPD}{\mathrm{ELPD}}
\newcommand{\avar}{\operatorname{avar}}

\newcommand{\psif}{\psi}
\newcommand{\PsiPop}{\Psi}
\newcommand{\PsiEmp}{\Psi_n}

\newcommand{\astar}{a^\ast}
\newcommand{\anstar}{a_n^\ast}
\newcommand{\aopt}{a_0^\ast}

\newcommand{\Apsi}{A}
\newcommand{\Bpsi}{B}

\usepackage{xcolor}
\usepackage{natbib}
\usepackage{hyperref}

\newcommand{\pii}{p_i}
\newcommand{\qia}{q_i(\cdot;a)}

\newcommand{\Geps}{G_\epsilon}
\newcommand{\Gept}{G_{\epsilon,\tau}}

\newcommand{\mzero}{m_0}
\newcommand{\meps}{m_\epsilon}
\newcommand{\mH}{m_H}

\newcommand{\Bzero}{B_0}
\newcommand{\Bstar}{B_\ast}

\newtheorem{theorem}{Theorem}
\newtheorem{lemma}{Lemma}
\newtheorem{corollary}{Corollary}

\theoremstyle{remark}
\newtheorem{remark}{Remark}

\title{Conformity-Based Bayesian Projective
Prediction}

\author{Arkaprava Roy$^1$, Malay Ghosh$^2$\\ $^1$ Department of Biostatistics, $^2$ Department of Statistics, University of Florida}

\begin{document}

\maketitle

\begin{abstract}
We propose a general robust prediction framework, termed
conformity-based projective prediction (CPP), that integrates Bayesian
predictive modeling with ideas from conformity-based conformal prediction. Rather than
assessing conformity through residual-based scores, the CPP criterion
defines conformity distributionally: a candidate value for a future
response is considered conforming to the extent that its inclusion in
the data leaves the leave-one-out predictive distributions of the
observed responses undisturbed. The framework requires only that the
leave-one-out and swapped predictive distributions are available in
closed form and that the swapped predictive mean is differentiable in
the candidate value. Under these conditions, we establish a general
bounded-influence proposition and a general local convexity lemma,
and prove that CPP dominates any plug-in predictor with unbounded
influence in asymptotic variance under $\epsilon$-contamination models.
When the posterior mean is linear in the observations --- as in
Gaussian linear models, basis-expansion regression, and Gaussian
process regression --- the swapped predictive mean is affine in the
candidate value, yielding closed-form or one-dimensional optimization
solutions and an efficient rank-two computational update; all general
theoretical results specialize to explicit corollaries in this setting.
Simulation experiments and two data analyses under the Gaussian linear
model illustrate the finite-sample advantages of the proposed method,
confirming the theoretical predictions across contamination levels,
sample sizes, and predictor dimensions.
\end{abstract}


\section{Introduction}

Prediction under uncertainty is a central problem in statistics, with
applications ranging from clinical decision-making and genomics to
environmental monitoring and econometrics. A recurring challenge is that
prediction methods trained on a finite sample must extrapolate to future observations whose distribution may not be fully specified, and must do so
in a way that provides reliable uncertainty quantification rather than
merely a point estimate. Two broad traditions address this challenge: the
frequentist conformal prediction framework, which provides finite-sample
coverage guarantees under minimal assumptions, and the Bayesian predictive
framework, which uses a probabilistic model to construct a full predictive
distribution. Each tradition has well-documented strengths and limitations,
and the relationship between them remains an active area of methodological
research.

Conformal prediction was introduced by \citet{Gammerman1998} and
developed systematically by \citet{Vovk2005}. The central idea is to
assess how well a candidate future observation conforms with the
previously observed data, relative to a user-chosen conformity measure.
Formally, let $y_1,\ldots,y_n$ denote the realizations of an exchangeable
sequence $Y_1,\ldots,Y_n$. For a candidate value $y$ at a future
covariate $x_{n+1}$, one computes conformity scores
\[
c_i(y) = C\bigl((x_1,y_1),\ldots,(x_{i-1},y_{i-1}),(x_{n+1},y),(x_{i+1},y_{i+1}),\ldots,(x_n,y_n)\bigr), \quad i=1,\ldots,n+1,
\]
and includes $y$ in the prediction region $R_{n+1}$ if
\[
\#\{i=1,\ldots,n+1 : c_i(y) \le c_{n+1}(y)\} > k,
\qquad
k = \lfloor (n+1)\alpha \rfloor.
\]
Exchangeability ensures that the conformity scores
$c_1,\ldots,c_{n+1}$ are themselves exchangeable, which underpins the
finite-sample validity of the resulting region: the coverage probability
satisfies $\mathrm{P}(Y_{n+1} \in R_{n+1}) \ge \alpha$, with equality
when $(n+1)\alpha$ is an integer \citep{Vovk2005, Shafer2008}. This
guarantee is distribution-free, requiring only exchangeability of the
extended sequence $(Y_1,\ldots,Y_{n+1})$, and holds regardless of the
dimension of the covariates or the complexity of the prediction model.

The appeal of this framework has stimulated a rich body of subsequent
work. \citet{Lei2014} and \citet{Sadinle2019} studied the statistical
efficiency of conformal regions and their connections to nonparametric
regression and set-valued classification; \citet{Vovk2016}
extended the framework to predictive distribution estimation;
\citet{Barber2021} introduced jackknife$+$ and cross-conformal
procedures that reduce the computational burden of full conformal
prediction while preserving approximate coverage; \citet{Romano2019}
proposed conformalized quantile regression as an efficient approach to
regression intervals; and \citet{Angelopoulos2023} provided a unified
treatment of conformal risk control beyond coverage. In a parallel
development, \citet{Tibshirani2019} proposed conformal prediction under
covariate shift via weighted conformity scores, relaxing the strict
exchangeability requirement. A comprehensive review of these developments
is given by \citet{Shafer2008} and, more recently, by
\citet{Angelopoulos2023}. The distribution-free predictive inference
framework of \citet{Lei2018} provides a unified treatment of regression
prediction that connects the conformal and classical statistical
perspectives.

The present work is inspired by this conformity-based prediction framework to ask a related but distinct question. Here, we ask whether the
distributional machinery of Bayesian predictive modeling can be used to
define a \emph{different} kind of predictive criterion, one that is
informed by the spirit of conformal thinking but operates in a
fundamentally different mode.

This distinction is worth drawing carefully. Standard conformal
prediction is primarily a framework for \emph{set-valued} prediction:
its output is a region $R_{n+1}$ with a finite-sample coverage
guarantee, and its appeal lies precisely in the distribution-free
character of that guarantee. The approach we develop here does not
compete with this goal. It produces a \emph{point predictor} rather
than a prediction region, it operates within a Bayesian parametric
model, and it makes no claim to distribution-free validity. What it
shares with conformal prediction is a structural idea: that a sensible
prediction for $y_{n+1}$ should be one whose inclusion in the dataset
is, in a well-defined distributional sense, quite satisfactory.

As an alternative to residual-based conformity scores, which underpin
most practical conformal methods, standard approaches discard a
substantial portion of the distributional information available in a
fitted model. In most applications, the conformity score takes the form
$c_i = |y_i - \hat{f}(x_i)|$, where $\hat{f}$ is a point estimator of
the regression function. When the fitted model is Bayesian or otherwise
probabilistic, however, the predictive distribution at each training
point encodes richer information about how well the observation conforms
to the data than the residual alone. This observation does not identify
a deficiency in conformal prediction, which does not require or assume
any particular model; it simply suggests that a predictive
density-based approach following \citet{hoff2023bayes} and
\citet{Fong2021} might exploit this information in a different way.

The second observation concerns robustness. Conformal prediction
inherits the robustness properties of the underlying conformity
measure, which in turn reflects the fitted model. When the training
data contain outliers or high-leverage observations, residual-based
scores derived from non-robust estimators can be distorted, leading to
degraded prediction efficiency even when nominal coverage is maintained
\citep{Hampel1974, Huber1981, Hampel1986}. Robust regression estimators
can be substituted as the underlying model \citep{Ronchetti1994}, but
this leaves the question of how to build robustness directly into a
\emph{distributional} notion as an avenue for investigation. The present
paper addresses this question from within the Bayesian predictive
framework, using divergence measures whose score functions are
redescending \citep{Basu1998, GhoshBasu2016}.

The framework we propose, which we term \emph{conformity-basedprojective
prediction}, integrates the Bayesian predictive method with conformal
thinking. Rather than defining conformity through residuals derived
from a point estimator, we define it distributionally: a candidate
value $a$ for $y_{n+1}$ is considered conforming to the extent that
its inclusion in the data leaves the leave-one-out predictive
distributions of the observed responses undisturbed. This leads to a
projection criterion in which the predicted value $a^\ast$ minimizes a
divergence between two families of Gaussian predictive laws, namely,
(1) the leave-one-out predictive distributions of $y_i$ given the
remaining training data, and (2) the corresponding distributions
obtained by swapping $(x_i,y_i)$ with the candidate pair $(x_{n+1}, a)$.
The resulting optimization problem is tractable in closed form under
the Gaussian linear model and admits natural extensions to
nonparametric regression via basis expansions and Gaussian processes.

The framework connects to several strands of the existing literature.
Within the Bayesian predictive tradition, it is related to the use of
leave-one-out cross-validation for model assessment \citep{Vehtari2017}
and to the literature on M-estimation and predictive projections
\citep{Gneiting2007, Dawid1984}. The divergence measures we
employ, the squared Hellinger distance and the density power
divergence of \citet{Basu1998}, are chosen for their redescending
score properties, which yield bounded influence functions and, hence,
resistance to gross-error contamination. This connects the framework
to the classical robust statistics program of \citet{Hampel1986}
and \citet{Huber1981}, and in particular to the minimum Hellinger
distance and minimum divergence estimation literature
\citep{Beran1977, Tamura1986, Simpson1987, Lindsay1994}. The
optimality theory we develop, including an asymptotic variance
comparison under $\epsilon$-contamination models, formalizes the sense
in which conformity-basedprojective prediction dominates plug-in prediction
when the data are contaminated.

Several features of the proposed framework are worth highlighting.
First, unlike standard conformal prediction, our method produces a
\emph{point predictor} rather than a prediction region; it is thus
complementary to, rather than a replacement for, conformal regions.
Second, the framework is explicitly Bayesian in its use of predictive
distributions, but the divergence criterion introduces a frequentist
robustness discipline that guards against model misspecification.
Third, because the optimization criterion depends on $a$ only through
the mean of the swapped predictive distribution---which is affine in
$a$ under Gaussian models---the method is computationally efficient:
for the log-Bhattacharyya criterion a closed-form solution is
available, while for Hellinger and density power divergence criteria,
the problem reduces to a one-dimensional optimization over a compact
interval regardless of $p$ or $n$ (Corollary~\ref{cor:1d}); the
required coefficients can be computed at $O(np^2)$ cost via a
rank-two Woodbury update (Corollary~\ref{cor:rank2}). Fourth, the
framework extends naturally to nonparametric settings via finite basis
approximations and Gaussian process priors \citep{Rasmussen2006},
and all theoretical guarantees---bounded influence, asymptotic
normality, and variance dominance under contamination---inherit
verbatim to both nonparametric settings
(Corollary~\ref{cor:nonparametric}).

The remainder of the paper is organized as follows.
Section~\ref{sec:model} contains the complete theoretical development,
structured general-first. Section~\ref{sec:general} defines the CPP
criterion abstractly and introduces conditions~(C1)--(C2) for
tractability, distinguishing the general differentiable case from
the affine special case that arises under Gaussian models.
The general robustness and asymptotic theory ---
bounded influence (Section~\ref{sec:scores}),
local convexity (Section~\ref{sec:convexity}),
asymptotic distribution and variance dominance
(Section~\ref{sec:asymptotics}), and predictive comparison under
contamination (Section~\ref{sec:contamination}) --- are each stated
at the general level requiring only (C1)--(C2).
Section~\ref{sec:loo} then verifies these conditions under the
Gaussian linear model, establishes the affine structure as a
consequence of posterior linearity, and collects all model-specific
corollaries of the preceding general results.
Section~\ref{sec:inheritance} states a formal corollary showing all
results carry over verbatim to nonparametric models.
Section~\ref{sec:nonparametric} verifies conditions~(C1)--(C2) and
the affine structure for basis-expansion and Gaussian process
regression.
Section~\ref{sec:unknown-variance} addresses implementation under
unknown variance in the Gaussian linear model setting, proposing two
plug-in strategies and analyzing their contamination bias; the
approach generalizes to other parametric models satisfying (C1) and
(C2). These implementation details, along with simulation experiments
(Section~\ref{sec:simulations}) and two real-data illustrations
(Section~\ref{sec:data}), are organized under a single section on
numerical illustration under the linear model. Section~\ref{sec:discussion}
concludes with a discussion of directions for future work, including
extensions to generalized linear models. Proofs of all major results
are collected in the Appendix.

\section{conformity-basedprojective prediction}
\label{sec:model}

\subsection{General framework}
\label{sec:general}

We propose a prediction framework that defines conformity
distributionally rather than through residuals. Let
$\mathcal{D}_n=\{(\bx_i,y_i)\}_{i=1}^n$ be the observed data and
let $(\bx_{n+1},a)$ denote a candidate observation, where $a\in\R$
is the value to be predicted.

\paragraph{Leave-one-out predictive distribution.}
For each $i=1,\ldots,n$, let $\pii$ denote the predictive distribution
of $y_i$ given all data except the $i$th observation:
\[
\pii = p\bigl(y_i \mid \mathcal{D}_n\setminus\{(\bx_i,y_i)\}\bigr).
\]
This is the LOO predictive distribution of $y_i$ under the fitted
model; it characterizes how the $i$th observation looks relative to
the rest of the data.

\paragraph{Swapped predictive distribution.}
Now replace the $i$th observation with the candidate pair
$(\bx_{n+1},a)$ and define
\[
\qia = p\bigl(y_i \mid
  \mathcal{D}_n\setminus\{(\bx_i,y_i)\}\cup\{(\bx_{n+1},a)\}
\bigr).
\]
This is the predictive distribution of $y_i$ after swapping out
$({\bx_i},y_i)$ and inserting the candidate $({\bx_{n+1}},a)$.
If the candidate value $a$ is a plausible response at $\bx_{n+1}$,
its inclusion should not substantially alter the predictive
distributions of the remaining observations; the two families
$\{\pii\}$ and $\{\qia\}$ should be close.

\paragraph{CPP criterion.}
This conformity idea is formalized by the
\emph{conformity-basedprojective prediction} (CPP) criterion:
\begin{equation}
\label{eq:cpp}
\astar = \arg\min_{a\in\R}\;
J(a) := \sum_{i=1}^n D\!\bigl(\pii,\,\qia\bigr),
\end{equation}
where $D(\cdot,\cdot)$ is a divergence between probability
distributions. The predicted value $\astar$ is the candidate that
minimizes the total divergence between the LOO and swapped predictive
families; it is the value whose inclusion in the data is most
``conforming'' in a distributional sense.

\paragraph{Conditions for tractability.}
The general criterion~\eqref{eq:cpp} is tractable whenever two
conditions hold:
\begin{enumerate}
\item[(C1)] \textit{Closed-form predictive distributions.}
  Both $\pii$ and $\qia$ are available analytically, so that
  $D(\pii,\qia)$ can be evaluated in closed form.
\item[(C2)] \textit{Differentiable swapped mean.}
  The mean of $\qia$ is a known, differentiable function of $a$,
  denoted $\mu_{1i}(a) := E_{q_i(\cdot;a)}[y_i]$, with derivative
  $\partial_a\mu_{1i}(a)$ available in closed form.
\end{enumerate}
Under these conditions, defining
\[
\mu_{2i} := E_{p_i}[y_i],
\qquad
\Delta_i(a) := \mu_{2i} - \mu_{1i}(a),
\]
the per-term divergence $D_i(a) = D(\pii,\qia)$ depends on $a$ only
through $\Delta_i(a)$, and the objective $J(a)=\sum_i D_i(a)$ is a
differentiable function of the single variable $a$ regardless of $p$
or $n$.

\paragraph{The affine special case.}
For any Bayesian model in which the posterior mean is linear in the
observations $\by$, the swapped mean $\mu_{1i}(a)$ is affine in $a$:
\[
\mu_{1i}(a) = c_i + d_i a,
\qquad
\partial_a\mu_{1i}(a) = d_i \quad \text{(constant)},
\]
for model-specific scalars $c_i$ and $d_i$. This affine structure
is verified for the Gaussian linear model in Section~\ref{sec:loo}
and for basis-expansion and Gaussian process regression in
Section~\ref{sec:nonparametric}. It yields two additional tractability
benefits beyond (C1)--(C2): (i) the minimizer of $J$ under Hellinger
or DPD is found by solving a single scalar equation
(Corollary~\ref{cor:1d}), and (ii) the scalars $c_i$, $d_i$ for all
$i$ can be computed from a single matrix inversion via a rank-two
Woodbury update at cost $O(np^2)$ (Corollary~\ref{cor:rank2}).
These benefits are consequences of the affine structure and do not
hold in general.

\paragraph{Model classes.}
Conditions~(C1) and~(C2) are satisfied, and the affine special case
holds, for any Bayesian model with a Gaussian likelihood and conjugate
prior, in which the posterior mean is linear in $\by$.
Section~\ref{sec:loo} verifies this for the Gaussian linear model;
Section~\ref{sec:nonparametric} verifies it for basis-expansion and
Gaussian process regression. For models outside this class, such as
generalized linear models, (C1) and (C2) may still hold under
approximation, but the affine structure will generally not, precluding
the rank-two update and closed-form 1D solution (see
Section~\ref{sec:discussion}).

\paragraph{Choice of divergence.}
The divergence $D$ in~\eqref{eq:cpp} determines the robustness
properties of $\astar$. We consider three choices:
\begin{itemize}
  \item \textit{Log-Bhattacharyya:} $D(\pii,\qia)=-\log\BC_i(a)$,
    where $\BC_i$ is the Bhattacharyya coefficient. This yields
    a closed-form solution under Gaussian predictives but has an
    unbounded score function and is not robust to gross outliers
    (Section~\ref{sec:projection}).
  \item \textit{Squared Hellinger distance:}
    $D(\pii,\qia)=H^2(\pii,\qia)$. The score function is
    redescending, yielding a bounded influence function
    (Section~\ref{sec:scores}).
  \item \textit{Density power divergence} \citep{Basu1998} with
    parameter $\alpha>0$: $D(\pii,\qia)=\mathrm{DPD}_\alpha(\pii,\qia)$.
    The score function is also redescending; $\alpha$ controls the
    trade-off between robustness and efficiency
    (Section~\ref{sec:scores}).
\end{itemize}
The Hellinger and DPD criteria are the primary proposals of this
paper. The log-Bhattacharyya criterion is included as a tractable
baseline that clarifies why redescending scores are necessary for
robustness.

\subsection{Score functions and bounded influence}
\label{sec:scores}

We establish bounded influence at the level of the general CPP
criterion~\eqref{eq:cpp}, under conditions~(C1) and~(C2) of
Section~\ref{sec:general}, requiring only that the per-term scores
satisfy a structural condition. The explicit Gaussian Hellinger and
DPD instantiations are collected in Section~\ref{sec:loo} as
corollaries of the general results here.

Let $T(F)$ denote the population functional defined by
$\PsiPop(a;F):=\E_F[\psif(Y;a)]=0$. Its influence function is
\[
\IF(z;T,F)
=
-
\left(
\frac{\partial}{\partial a}\E_F[\psif(Y;a)]
\right)^{-1}_{a=T(F)}
\psif(z;T(F)).
\]

\begin{assumption}[Redescending score structure]
\label{ass:redesc}
The per-term score of $J(a)=\sum_{i=1}^n D_i(a)$ takes the form
\[
\psif_i(a)
:=
\frac{\partial}{\partial a}D_i(a)
=
-\,\dot\mu_{1i}(a)\,g_i\!\bigl(\Delta_i(a)\bigr),
\qquad
\dot\mu_{1i}(a) := \partial_a\mu_{1i}(a),
\]
where $g_i:\R\to\R$ is an odd function satisfying
$\sup_{\delta\in\R}|g_i(\delta)|\le M_i < \infty$
and $g_i(\delta)\to 0$ as $|\delta|\to\infty$.
In the affine special case $\mu_{1i}(a)=c_i+d_i a$,
$\dot\mu_{1i}(a)=d_i$ is constant and this reduces to
$\psif_i(a)=-d_i\,g_i(\Delta_i(a))$.
\end{assumption}

\begin{proposition}[General bounded influence]
\label{prop:bounded-inf}
Under Assumption~\ref{ass:redesc}, suppose $|\dot\mu_{1i}(a)|\le L_i$
uniformly in $a$ for each $i$. Then the aggregate score
$\psif(a)=\sum_i\psif_i(a)$ is bounded:
$|\psif(a)|\le\sum_i L_i M_i<\infty$ for all $a\in\R$.
Consequently, if $\partial_a\E_F[\psif(Y;a)]$ is finite and nonzero
at $a=T(F)$, then $\sup_z |\IF(z;T,F)| < \infty$.
In the affine case $\dot\mu_{1i}(a)=d_i$, the bound becomes
$\sum_i|d_i|M_i$.
\end{proposition}

\begin{proof}
$|\psif_i(a)|=|\dot\mu_{1i}(a)||g_i(\Delta_i(a))|\le L_i M_i$.
Summing over $i$ gives the score bound. The influence function
bound follows from the standard formula.
\end{proof}

\subsection{Local convexity and identification}
\label{sec:convexity}

We establish local strong convexity of $J$ under the general CPP
criterion, requiring only that each $D_i$ satisfies a sign condition
on its second derivative. The explicit convexity regions for the
Gaussian Hellinger and DPD divergences are collected in
Section~\ref{sec:loo}.

\begin{lemma}[General local strong convexity]
\label{lem:convexity}
Let $J(a)=\sum_{i=1}^n D_i(a)$ where $D_i(a)=D_i(\Delta_i(a))$
depends on $a$ only through $\Delta_i(a)=\mu_{2i}-\mu_{1i}(a)$,
with $\dot\mu_{1i}(a):=\partial_a\mu_{1i}(a)\neq 0$ for at least
one $i$. Let $\mathcal{R}_i:=\{\delta:D_i''(\delta)>0\}$.
If $\Delta_i(\astar)\in\mathcal{R}_i$ for each $i$ and
$\sum_i \dot\mu_{1i}(\astar)^2 D_i''(\Delta_i(\astar))\ge c>0$,
then $J''(\astar)\ge c>0$.
In the affine special case $\mu_{1i}(a)=c_i+d_i a$,
$\dot\mu_{1i}(\astar)=d_i$ and the condition reduces to
$\sum_i d_i^2 D_i''(\Delta_i(\astar))\ge c>0$.
\end{lemma}

\begin{proof}
By the chain rule,
$J'(a)=\sum_i D_i'(\Delta_i(a))\cdot(-\dot\mu_{1i}(a))$ and
$J''(a)=\sum_i\bigl[\dot\mu_{1i}(a)^2 D_i''(\Delta_i(a))
-\ddot\mu_{1i}(a) D_i'(\Delta_i(a))\bigr]$.
At a minimizer $\astar$ the first-order condition gives
$\sum_i D_i'(\Delta_i(\astar))\dot\mu_{1i}(\astar)=0$.
For the affine case $\ddot\mu_{1i}\equiv 0$, so
$J''(\astar)=\sum_i d_i^2 D_i''(\Delta_i(\astar))\ge c$.
In general, local strong convexity holds whenever the dominant
term $\sum_i\dot\mu_{1i}(\astar)^2 D_i''(\Delta_i(\astar))\ge c>0$
and the curvature correction $\sum_i\ddot\mu_{1i}(\astar)D_i'(\Delta_i(\astar))$
is small relative to $c$.
\end{proof}

Although $J$ is not globally convex, Lemma~\ref{lem:convexity}
guarantees local strong convexity near any minimizer $\astar$ at
which predictive discrepancies are small. This ensures local
identification and the non-vanishing of $J''(\astar)$, required for
the bounded-influence conclusion of Proposition~\ref{prop:bounded-inf}
and for the asymptotic theory below.

\paragraph{Local sensitivity.}
Differentiating $J'(\astar;y)=0$ with respect to $y_j$ gives
$d\astar/dy_j = -\partial_{y_j}J'|_{a=\astar}/J''(\astar)$.
Lemma~\ref{lem:convexity} ensures $J''(\astar)\ge c>0$; if
$\partial_{y_j}J'$ is uniformly bounded near $\astar$, then
$|d\astar/dy_j|\le\sup|\partial_{y_j}J'|/c<\infty$,
so $\astar(y)$ is locally Lipschitz-stable.
By Proposition~\ref{prop:bounded-inf}, the Hellinger and DPD CPP
predictors additionally have globally bounded influence functions.

\subsection{Asymptotic theory}
\label{sec:asymptotics}

We now state the main asymptotic results for the CPP predictor. All
results hold for any model satisfying conditions~(C1) and~(C2) of
Section~\ref{sec:general}; the Gaussian linear model of
Section~\ref{sec:loo} is treated as a special case via corollaries.
All proofs are deferred to Appendix~\ref{app:proofs}.

\begin{proposition}[Asymptotic distribution]
\label{prop:asymp-dist}
Let $\anstar$ solve $\PsiEmp(a)=n^{-1}\sum_{i=1}^n \psif(Y_i;a)=0$, and
let $\aopt$ be the unique solution of $\PsiPop(a)=\E[\psif(Y;a)]=0$.
Suppose:
\begin{enumerate}
\item[(A1)] $\Apsi:=\E[\partial_a\psif(Y;\aopt)]\neq 0$.
\item[(A2)] For every $\epsilon>0$, there exists $\delta>0$ near $a^*_0$, $\{\sup_{u\in [\aopt-\delta,\aopt+\delta]} |\partial_a\psif(Y;a)_{a=u}-\partial_a\psif(Y;a)_{a=a_0^*}|>\epsilon\}\rightarrow0$ as $n\rightarrow\infty$;
\item[(A3)] $\E[\psif(Y;\aopt)^2]<\infty$.
\end{enumerate}
Then
\[
\sqrt{n}(\anstar-\aopt)
\dto
N\!\left(0,\,\frac{\Bpsi}{\Apsi^2}\right),
\qquad
\Bpsi=\Var\!\bigl(\psif(Y;\aopt)\bigr).
\]
\end{proposition}

\begin{remark}
Assumption~(A1) requires the sensitivity $\Apsi$ to be nonzero.
Under the Hellinger or DPD score, local strong convexity of $J$ at
$\astar$ (Lemma~\ref{lem:convexity}) implies
$J''(\astar)>0$, which is equivalent to $\Apsi\neq 0$.
Thus~(A1) is verified whenever the conditions of Lemma~\ref{lem:convexity}
hold, and need not be treated as an independent assumption.
\end{remark}

\begin{proposition}[Asymptotic variance comparison]
\label{prop:avar}
Let $\hat{f}_n(\bx_{n+1})$ be any plug-in predictor satisfying
$\sqrt{n}(\hat{f}_n(\bx_{n+1})-\mzero)\dto N(0,V_0)$ for some
$V_0>0$ under the true model $F_0$.
Under the assumptions of Proposition~\ref{prop:asymp-dist},
\[
\avar(\anstar)=\frac{\Bpsi}{\Apsi^2},
\qquad
\avar\bigl(\hat{f}_n(\bx_{n+1})\bigr)=V_0.
\]
\end{proposition}

The following theorem formalizes CPP dominance over any plug-in
predictor whose score has unbounded influence. The key inputs are:
(i) the uniform boundedness of the robust score
(Proposition~\ref{prop:bounded-inf}), and (ii) the uniform positivity
of the sensitivity $A^\ast(\epsilon,\tau)$
(Proposition~\ref{prop:finite-range} below).

\begin{proposition}[Finite-range identification bound]
\label{prop:finite-range}
Fix $T>0$ and consider $\Gept=(1-\epsilon)F_0+\epsilon H_\tau$,
$0<\tau\le T$. Let $a_{\epsilon,\tau}^\ast$ be the unique root of
$\Psi^\ast(a;\Gept)=0$. Assume: the contaminated roots remain in a
compact interval $K_T$; and $\Delta_i(a_{\epsilon,\tau}^\ast)\in\mathcal{R}_i$
throughout $K_T$ (Lemma~\ref{lem:convexity}).

Then there exist $0<\underline{A}_T\le\overline{A}_T<\infty$ such that
\[
\underline{A}_T
\le
\bigl|\partial_a\Psi^\ast(a;\Gept)\big|_{a=a_{\epsilon,\tau}^\ast}\bigr|
\le
\overline{A}_T,
\qquad
0<\tau\le T.
\]
\end{proposition}

\begin{remark}
The local convexity condition in Proposition~\ref{prop:finite-range}
is guaranteed by Lemma~\ref{lem:convexity} whenever the minimizer
lies in the local convexity region. Thus $\underline{A}_T>0$ follows
without further assumptions beyond those of Lemma~\ref{lem:convexity}.
\end{remark}

\begin{theorem}[Asymptotic variance dominance under contamination]
\label{thm:dominance}
Let $\Gept=(1-\epsilon)F_0+\epsilon H_\tau$ with $0<\epsilon<1$.
Let $\anstar$ satisfy Assumption~\ref{ass:redesc} with uniform score
bound $M=\sum_i L_i M_i$ (Proposition~\ref{prop:bounded-inf}), and
let $\hat{f}_n(\bx_{n+1})$ be any plug-in predictor whose score
function $\psif_0$ satisfies
$\E_{H_\tau}[\psif_0(Y)^2]\to\infty$ as $\tau\to\infty$.
Let $V^\ast(\epsilon,\tau)$ and $V_0(\epsilon,\tau)$ denote the
respective asymptotic variances under $\Gept$. Then:
\begin{enumerate}
\item[(i)] $\sup_{0<\tau\le T}V^\ast(\epsilon,\tau)\le M^2/\underline{A}_T^2<\infty$
  for every finite $T>0$.
\item[(ii)] $V_0(\epsilon,\tau)\to\infty$ as $\tau\to\infty$.
\item[(iii)] There exists $\tau_0<\infty$ such that
  $V^\ast(\epsilon,\tau)<V_0(\epsilon,\tau)$ for all $\tau\ge\tau_0$.
\end{enumerate}
\end{theorem}

Part~(i) of Theorem~\ref{thm:dominance} uses the uniform bound
$\underline{A}_T>0$, whose validity for the Hellinger and DPD
criteria follows from Proposition~\ref{prop:finite-range} and
Lemma~\ref{lem:convexity}.

\begin{theorem}[Predictive consistency]
\label{thm:consistency}
Let $\anstar=\arg\min_{a\in\mathcal{A}}J_n(a)$, where
$J_n(a)=n^{-1}\sum_i D(p_i,q_i(\cdot;a))$ and $\mathcal{A}$ is compact.
If $J(a)=E[D(p(Y),q(Y;a))]$ has a unique minimizer $\aopt$, and
$\sup_{a\in\mathcal{A}}|J_n(a)-J(a)|\pto 0$, then $\anstar\pto\aopt$.
Under correct specification and a divergence that is uniquely minimized
when the predictive mean matches the truth, $\aopt=\mzero:=E(Y\mid x_{n+1})$.
\end{theorem}

\begin{corollary}[Efficiency equivalence under correct specification]
\label{cor:clean}
Under correct specification ($\epsilon=0$), both $\anstar$ and
$\hat{f}_n(\bx_{n+1})$ are consistent for $\mzero=E(Y\mid\bx_{n+1})$
by Theorem~\ref{thm:consistency} and standard plug-in theory,
respectively. Their asymptotic variances satisfy
\[
\avar(\anstar) = \frac{B^\ast_0}{A^{\ast2}_0},
\qquad
\avar\bigl(\hat{f}_n(\bx_{n+1})\bigr) = V_0,
\]
where $A^\ast_0:=A^\ast|_{\epsilon=0}$ and $B^\ast_0:=B^\ast|_{\epsilon=0}$.
The ratio $B^\ast_0/A^{\ast2}_0$ may exceed $V_0$ (mild efficiency
loss under $F_0$), but Theorem~\ref{thm:dominance} guarantees this
is offset by unbounded variance gains under contamination. Under the
Gaussian linear model with OLS plug-in, $V_0=\bx_{n+1}^T\Sigma_\beta\bx_{n+1}$
(Corollary~\ref{cor:avar-ols}). This theoretical equivalence is
consistent with the simulation finding in Section~\ref{sec:simulations} under clear data.
\end{corollary}

\subsection{Predictive comparison under contamination}
\label{sec:contamination}

We compare expected log-predictive densities (ELPD) under
$\epsilon$-contamination for any predictive model in which ELPD is
locally quadratic in the prediction. The explicit Gaussian ELPD
formula is collected in Section~\ref{sec:loo}.

\paragraph{General contamination expansion.}
Consider $\Geps=(1-\epsilon)F_0+\epsilon H$, $0\le\epsilon<1$, with
contaminated predictors
$a_{0,\epsilon}=\aopt+\epsilon\Bzero(H)+o(\epsilon)$ and
$a_{\ast,\epsilon}=\astar+\epsilon\Bstar(H)+o(\epsilon)$,
where $\Bzero(H)=\E_H[\IF(Z;T,F_0)]$ and
$\Bstar(H)=\E_H[\IF(Z;T_\ast,F_0)]$.

\begin{proposition}[General ELPD comparison under contamination]
\label{prop:elpd-general}
Let the predictive model satisfy, for some strictly proper scoring
rule $S$,
\[
S(a;\Geps) \equiv -\lambda(\meps-a)^2 + \text{const}
\]
up to first order in $\epsilon$, where $\lambda>0$ and
$\meps=\E_{\Geps}(Y)=\mzero+\epsilon(\mH-\mzero)$.
Then, to first order in $\epsilon$,
\[
S(\astar;\Geps)-S(\aopt;\Geps)
\approx
\lambda\bigl[(\mzero-\aopt)^2-(\mzero-\astar)^2\bigr]
+
2\lambda\epsilon\bigl[
(\mzero-\aopt)\{(\mH-\mzero)-\Bzero(H)\}
-(\mzero-\astar)\{(\mH-\mzero)-\Bstar(H)\}
\bigr].
\]
For any plug-in predictor with unbounded influence, $|\Bzero(H)|$
may be arbitrarily large under heavy-tailed contamination. For the
Hellinger or DPD CPP predictor,
Proposition~\ref{prop:bounded-inf} gives $|\Bstar(H)|\le M$,
controlling the first-order contamination bias.
\end{proposition}

\begin{proof}
Substituting $\meps=\mzero+\epsilon(\mH-\mzero)$ and expanding
$(\meps-a_{0,\epsilon})^2$ and $(\meps-a_{\ast,\epsilon})^2$ to
first order in $\epsilon$ using
$a_{0,\epsilon}=\aopt+\epsilon\Bzero(H)+o(\epsilon)$ and
$a_{\ast,\epsilon}=\astar+\epsilon\Bstar(H)+o(\epsilon)$,
\begin{align*}
(\meps-a_{0,\epsilon})^2
&=
(\mzero-\aopt)^2
+2\epsilon(\mzero-\aopt)\{(\mH-\mzero)-\Bzero(H)\}+o(\epsilon),
\end{align*}
and similarly for $a_{\ast,\epsilon}$.
The scoring assumption $S(a;\Geps)\equiv-\lambda(\meps-a)^2+\text{const}$
gives $S(\astar;\Geps)-S(\aopt;\Geps)
=\lambda[(\meps-a_{0,\epsilon})^2-(\meps-a_{\ast,\epsilon})^2]$.
Substituting the expansions and collecting terms of order $1$ and
$\epsilon$ yields the stated expression. The bound on $\Bstar(H)$
follows from $|\IF(Z;T_\ast,F_0)|\le M$
(Proposition~\ref{prop:bounded-inf}).
\end{proof}

\subsection{Gaussian linear model: verification and special cases}
\label{sec:loo}

We verify conditions~(C1) and~(C2) of Section~\ref{sec:general}
under the Gaussian linear model, establishing the affine structure
$\mu_{1i}(a)=c_i+d_i a$ that underlies all subsequent corollaries.
We then collect the model-specific instantiations of the general
results from Sections~\ref{sec:scores}--\ref{sec:contamination}.
The LOO predictive distribution is derived first; the swapped
distribution and affine structure follow in Section~\ref{sec:swapped}.

Assume the Bayesian linear model
\[
\by \mid \bbeta \sim N(\bX\bbeta,\sigma^2 I),
\qquad
\bbeta \sim N(\bbeta_0,\sigma^2 \bV),
\]
where \(\bX\) denotes the design matrix, consisting either of the raw
covariates or the basis-expanded covariates in the semiparametric
formulation with \(p\) basis functions. We assume a fixed prior
specification and known \(\sigma^2\).

Let
\[
\bA = \bX^T \bX + \bV^{-1},
\qquad
\bba = \bX^T \by + \bV^{-1}\bbeta_0,
\]
so that the posterior mean based on the full data is
\[
\widehat\bbeta = \bA^{-1}\bba.
\]

When the $i$th observation is removed,
\[
E(\bbeta \mid \by_{(-i)})
=
(\bA-\bx_i\bx_i^T)^{-1}(\bba-\bx_i y_i).
\]

Using the Sherman--Morrison identity,
\[
(\bA-\bx_i\bx_i^T)^{-1}
=
\bA^{-1}
+
\frac{\bA^{-1}\bx_i \bx_i^T \bA^{-1}}
{1-\bx_i^T \bA^{-1}\bx_i}.
\]

Substituting this identity and defining $\ell_i = \bx_i^T \bA^{-1}\bx_i$ yields
\[
\begin{aligned}
E(\bbeta\mid \by_{(-i)})
&=
\left(
\bA^{-1}
+
\frac{\bA^{-1}\bx_i \bx_i^T \bA^{-1}}
{1-\bx_i^T \bA^{-1}\bx_i}
\right)(\bba-\bx_i y_i) \\
&=
\widehat\bbeta
-
\frac{\bA^{-1}\bx_i\,(y_i-\bx_i^T\widehat\bbeta)}
{1-\ell_i}.
\end{aligned}
\]

Therefore
\[
E(y_i\mid \by_{(-i)})
=
\bx_i^T\widehat\bbeta
-
\frac{\ell_i}{1-\ell_i}\,(y_i-\bx_i^T\widehat\bbeta)=\frac{\bx_i^T\widehat\bbeta-\ell_i y_i}{1-\ell_i}.
\]

The posterior covariance with the $i$th observation removed is
\[
\mathrm{Var}(\bbeta\mid \by_{(-i)})
=
\sigma^2 (\bA-\bx_i\bx_i^T)^{-1},
\]
and therefore
$\mathrm{Var}(y_i\mid \by_{(-i)})
= \sigma^2 + \bx_i^T\mathrm{Var}(\bbeta\mid \by_{(-i)})\bx_i
= \sigma^2/(1-\ell_i)$.
Thus,
\[
y_i\mid \by_{(-i)}
\sim
N\!\left(m_{2i},\, s_{2i}^2\right),
\qquad
m_{2i}
=
\frac{\bx_i^T\widehat\bbeta-\ell_i y_i}{1-\ell_i},
\quad
s_{2i}^2
=
\frac{\sigma^2}{1-\ell_i}.
\]

\subsubsection{Swapped predictive distribution}
\label{sec:swapped}

Under the same Gaussian linear model, we now derive $\qia$,
verifying condition~(C2) and establishing the affine structure
$m_{1i}(a)=c_i+d_i a$ that underpins all subsequent results.

Augment the leave-one-out sample by a new point $(\bx_{n+1},a)$ and
define
\[
\bX^{(+)}_{-i}
=
\begin{pmatrix}
\bX_{-i}\\
\bx_{n+1}^T
\end{pmatrix},
\qquad
\by^{(+)}_{-i}
=
\begin{pmatrix}
\by_{-i}\\
a
\end{pmatrix}.
\]
Let
\[
\bA^{(+)}_{-i}
=
(\bX^{(+)}_{-i})^T \bX^{(+)}_{-i} + \bV^{-1}
=
\bA - \bx_i\bx_i^T + \bx_{n+1}\bx_{n+1}^T,
\qquad
\bba^{(+)}_{-i}
=
(\bX^{(+)}_{-i})^T \by^{(+)}_{-i} + \bV^{-1}\bbeta_0,
\]
and let $\widehat\bbeta^{(+)}_{-i} = (\bA^{(+)}_{-i})^{-1}\bba^{(+)}_{-i}$.
Then
\[
\bbeta \mid \{(\by_{-i},\bX_{-i}),(a,\bx_{n+1})\}
\sim
N\!\left(\widehat\bbeta^{(+)}_{-i},\ \sigma^2 (\bA^{(+)}_{-i})^{-1}\right),
\]
and therefore
\[
y_i \mid \{(\by_{-i},\bX_{-i}),(a,\bx_{n+1})\}
\sim
N\!\left(m_{1i}(a),\, s_{1i}^2\right),
\]
where
\[
m_{1i}(a)=\bx_i^T\widehat\bbeta^{(+)}_{-i},
\qquad
s_{1i}^2=\sigma^2\bigl(1+\delta_i\bigr),
\qquad
\delta_i=\bx_i^T (\bA^{(+)}_{-i})^{-1}\bx_i.
\]

Since
\[
\widehat\bbeta^{(+)}_{-i}
=
(\bA^{(+)}_{-i})^{-1}
\bigl(\bX_{-i}^T \by_{-i}+\bx_{n+1}a+\bV^{-1}\bbeta_0\bigr),
\]
the mean is affine in $a$:
\[
m_{1i}(a)=c_i+d_i a,
\]
where
\[
c_i
=
\bx_i^T(\bA^{(+)}_{-i})^{-1}
\bigl(\bX_{-i}^T \by_{-i}+\bV^{-1}\bbeta_0\bigr),
\qquad
d_i
=
\bx_i^T(\bA^{(+)}_{-i})^{-1}\bx_{n+1}.
\]
This affine structure is the key computational property that makes the
subsequent optimization tractable, and it is preserved identically in
the basis-expansion and Gaussian process extensions of Section~\ref{sec:nonparametric}.

\begin{corollary}[Efficient computation via rank-two update]
\label{cor:rank2}
The matrix $\bA^{(+)}_{-i}=\bA-\bx_i\bx_i^T+\bx_{n+1}\bx_{n+1}^T$
is a rank-two update of $\bA$. Applying the Woodbury identity
sequentially, first to remove $\bx_i\bx_i^T$ and then to add
$\bx_{n+1}\bx_{n+1}^T$, gives
\[
(\bA^{(+)}_{-i})^{-1}
=
\bB_i
-
\frac{\bB_i\bx_{n+1}\bx_{n+1}^T\bB_i}{1+\bx_{n+1}^T\bB_i\bx_{n+1}},
\qquad
\bB_i
:=
\bA^{-1}+\frac{\bA^{-1}\bx_i\bx_i^T\bA^{-1}}{1-\ell_i}.
\]
Consequently, the scalars $c_i$ and $d_i$ for all $i=1,\ldots,n$
can be computed from $\bA^{-1}$ alone at total cost $O(np^2)$,
since $\bA^{-1}$ is computed once at cost $O(p^3)$ and each
rank-two update requires only $O(p^2)$ operations. This avoids
the naive $O(np^3)$ cost of inverting $\bA^{(+)}_{-i}$ separately
for each $i$, making the method computationally efficient at scale.
\end{corollary}

\subsubsection{Conformity-based projection under the Gaussian model}
\label{sec:projection}

With conditions~(C1) and~(C2) verified, we now solve the general
CPP criterion~\eqref{eq:cpp} explicitly under the three divergence
choices introduced in Section~\ref{sec:general}.

\noindent
{\bf{(i) Log-Bhattacharyya-based solution}}

Consider the two Gaussian predictive distributions
$N(m_{1i}(a),s_{1i}^2)$ and $N(m_{2i},s_{2i}^2)$.
Let
\[
D_i = s_{1i}^2+s_{2i}^2
=
\sigma^2(1+\delta_i)+\frac{\sigma^2}{1-\ell_i}.
\]

For univariate Gaussian distributions, the Bhattacharyya coefficient \citep{bhattacharyya1943measure} is
\[
BC_i
=\int N^{1/2}(x|m_{1i},s_{1i}^2)N^{1/2}(x|m_{2i},s_{2i}^2)dx=
\sqrt{\frac{2 s_{1i}s_{2i}}{s_{1i}^2+s_{2i}^2}}
\exp\!\left(
-\frac{(m_{1i}(a)-m_{2i})^2}{4(s_{1i}^2+s_{2i}^2)}
\right).
\]

Hence
\[
\log BC_i
=
\frac12
\log\!\left(
\frac{2 s_{1i}s_{2i}}{s_{1i}^2+s_{2i}^2}
\right)
-
\frac{(c_i+d_i a-m_{2i})^2}{4D_i}.
\]

The first term is constant in $a$, so maximizing $\sum_{i=1}^n \log BC_i$
is equivalent to minimizing
\[
\sum_{i=1}^n \frac{(c_i+d_i a-m_{2i})^2}{4D_i}.
\]

The closed-form minimizer is
\[
\astar
=
\frac{
\sum_{i=1}^n
\dfrac{d_i}{D_i}
\left(
\dfrac{\bx_i^T\widehat\bbeta-\ell_i y_i}{1-\ell_i}
-c_i
\right)
}{
\sum_{i=1}^n \dfrac{d_i^2}{D_i}
},
\]
substituting $m_{2i}=(\bx_i^T\widehat\bbeta-\ell_i y_i)/(1-\ell_i)$.

\paragraph{The log-BC score is not redescending.}
Let $\Delta_i(a)=m_{2i}-(c_i+d_i a)$.
The Bhattacharyya coefficient between $\pi_i=N(\mu_i,\sigma_i^2)$ and
$q_i(\cdot;a)=N(c_i+d_i a,\sigma_{i,+}^2)$ is
\[
BC_i(a)
=
\sqrt{\frac{2\sigma_i\sigma_{i,+}}{\sigma_i^2+\sigma_{i,+}^2}}
\exp\!\left\{
-\frac{\Delta_i(a)^2}{4(\sigma_i^2+\sigma_{i,+}^2)}
\right\},
\]
so $-\log BC_i(a)=C_i+\Delta_i(a)^2/[4(\sigma_i^2+\sigma_{i,+}^2)]$.
Differentiating with respect to $a$,
\[
\psif_i^{(\mathrm{BC})}(a)
:=
\frac{\partial}{\partial a}\{-\log BC_i(a)\}
=
-\frac{d_i\,\Delta_i(a)}{2(\sigma_i^2+\sigma_{i,+}^2)},
\]
and $|\psif_i^{(\mathrm{BC})}(a)|\to\infty$ as $|\Delta_i(a)|\to\infty$.

\begin{remark}
The log-Bhattacharyya-based predictor is linear in $y_i$, and its score
is not redescending. It is therefore not robust to gross outliers.
\end{remark}

\noindent
{\bf{(ii) Hellinger and DPD-based solutions}}

For a more robust criterion, we define
\[
\astar=\arg\min_{a\in\R}\; J(a),
\qquad
J(a)=\sum_{i=1}^n D_i(a),
\]
where $D_i(a)$ is either the squared Hellinger distance or the density
power divergence (DPD) between $N(m_{2i},s_{2i}^2)$ and
$N(m_{1i}(a),s_{1i}^2)$.

Because $m_{1i}(a)=c_i+d_i a$, the objective depends on $a$ only
through $\Delta_i(a)=m_{2i}-(c_i+d_i a)$.

\begin{corollary}[One-dimensional score equation]
\label{cor:1d}
Under the Hellinger or DPD criterion, the first-order condition
$J'(\astar)=0$ reduces to a single scalar equation in the unknown
$a\in\R$:
\[
\sum_{i=1}^n \psif_i^{(\cdot)}(\astar) = 0,
\]
where $\psif_i^{(\cdot)}$ is the score function established for
the relevant divergence in Section~\ref{sec:scores}. Because
$\Delta_i(a)=m_{2i}-(c_i+d_i a)$
is affine in $a$, $J$ is a function of the single variable $a$
regardless of the dimension $p$ of the covariates or the sample size
$n$. The minimizer $\astar$ is therefore obtained by
one-dimensional numerical optimization --- for example, grid search
or bisection applied to $J'(a)=0$ --- over a compact interval
determined by the support of the predictive distributions.
\end{corollary}

\paragraph{Gaussian special cases of the general theory.}
The general results of Sections~\ref{sec:scores}--\ref{sec:contamination}
specialize to the Gaussian linear model as follows.

\begin{corollary}[Hellinger and DPD scores: Gaussian special case]
\label{cor:scores-gaussian}
Since $\mu_{1i}(a)=c_i+d_i a$ is affine, $\dot\mu_{1i}(a)=d_i$
(constant) and $L_i=|d_i|$. Assumption~\ref{ass:redesc} holds with:

\noindent
\textit{Hellinger distance.}
Let $C_i=\sqrt{2s_{2i}s_{1i}/S_i}$, $S_i=s_{2i}^2+s_{1i}^2$.
Then $g_i^{(H)}(\delta)=-C_i\delta/(2S_i)\exp(-\delta^2/(4S_i))$,
which is odd and redescending with $M_i^{(H)}=C_i/(2S_i)\sqrt{S_i/(2e)}$.

\noindent
\textit{Density power divergence} (parameter $\alpha>0$).
With $\kappa_i=\alpha/(2(s_{2i}^2+\alpha s_{1i}^2))>0$,
$g_i^{(\alpha)}(\delta)\propto\delta\exp(-\kappa_i\delta^2)$,
also odd and redescending with $M_i^{(\alpha)}\propto(2e\kappa_i)^{-1/2}$.

In both cases Proposition~\ref{prop:bounded-inf} gives
$\sup_z|\IF(z;T,F)|<\infty$ with bound $\sum_i|d_i|M_i$.
\end{corollary}

\begin{proof}
Both $g_i^{(H)}$ and $g_i^{(\alpha)}$ are computed by direct
differentiation of the respective divergence with respect to $a$,
using $\partial_a\Delta_i(a)=-d_i$ (affine case).
The boundedness of $|g_i(\delta)|$ follows since $|\delta|e^{-c\delta^2}$
is bounded for any $c>0$. Assumption~\ref{ass:redesc} is satisfied
with $L_i=|d_i|$, so Proposition~\ref{prop:bounded-inf} applies.
\end{proof}

\begin{remark}
The log-BC score $\psif_i^{(\mathrm{BC})}(a)
=-d_i\Delta_i(a)/[2(\sigma_i^2+\sigma_{i,+}^2)]$
is linear in $\Delta_i(a)$ and therefore unbounded, violating
Assumption~\ref{ass:redesc}. This is why the log-BC predictor
is not robust to gross outliers.
\end{remark}

\paragraph{Local versus global robustness.}
Proposition~\ref{prop:bounded-inf} establishes \emph{global} bounded
influence: the score is bounded for arbitrarily large discrepancies,
since $g_i(\delta)\to 0$ as $|\delta|\to\infty$ and $\dot\mu_{1i}$
is uniformly bounded. This contrasts with OLS, where
$\psif_0(y_j)\propto(y_j-x_j^T\beta)$ is unbounded. The local
Lipschitz stability derived from $J''(\astar)>0$
(Lemma~\ref{lem:convexity}) is a strictly weaker property.

\begin{corollary}[Convexity regions: Gaussian special case]
\label{cor:convexity-gaussian}
Since $\mu_{1i}(a)=c_i+d_i a$ is affine, $\dot\mu_{1i}=d_i$
and $\ddot\mu_{1i}=0$, so Lemma~\ref{lem:convexity} gives
$J''(\astar)=\sum_i d_i^2 D_i''(\Delta_i(\astar))$, with explicit regions:
\begin{enumerate}
\item[(i)] \textit{Hellinger:}
  $D_i''(\delta)>0$ iff $|\delta|<\sqrt{2S_i}$, so
  $\mathcal{R}_i=(-\sqrt{2S_i},\sqrt{2S_i})$.
\item[(ii)] \textit{DPD} ($\alpha>0$):
  $D_i''(\delta)>0$ iff $|\delta|<1/\sqrt{2\kappa_i}$, so
  $\mathcal{R}_i=(-1/\sqrt{2\kappa_i},\,1/\sqrt{2\kappa_i})$.
\end{enumerate}
In both cases $J''(\astar)\ge c>0$ whenever
$\Delta_i(\astar)\in\mathcal{R}_i$ for each $i$ and
$\sum_i d_i^2 D_i''(\Delta_i(\astar))\ge c>0$.
\end{corollary}

\begin{proof}
\textit{Hellinger.}
$H_i^2(\delta)=1-C_i e^{-\delta^2/(4S_i)}$, so
$H_i''(\delta)=C_i e^{-\delta^2/(4S_i)}(1/(2S_i)-\delta^2/(4S_i^2))$,
positive iff $|\delta|<\sqrt{2S_i}$.
\textit{DPD.}
$D_{\alpha,i}''(\delta)=2K_i\kappa_i e^{-\kappa_i\delta^2}(1-2\kappa_i\delta^2)$,
positive iff $|\delta|<1/\sqrt{2\kappa_i}$.
\end{proof}

\begin{corollary}[Asymptotic variance: OLS plug-in special case]
\label{cor:avar-ols}
Under the Gaussian linear model with plug-in predictor
$\hat{f}_n(\bx_{n+1})=\bx_{n+1}^T\widehat\bbeta_n$ and
$\sqrt{n}(\widehat\bbeta_n-\bbeta_0)\dto N(0,\Sigma_\beta)$,
Proposition~\ref{prop:avar} gives $V_0=\bx_{n+1}^T\Sigma_\beta\bx_{n+1}$.
\end{corollary}

\begin{proof}
The delta method applied to $g(\bbeta)=\bx_{n+1}^T\bbeta$ gives
$\sqrt{n}(\bx_{n+1}^T\widehat\bbeta_n-\bx_{n+1}^T\bbeta_0)
\dto N(0,\bx_{n+1}^T\Sigma_\beta\bx_{n+1})$,
which identifies $V_0=\bx_{n+1}^T\Sigma_\beta\bx_{n+1}$ in
Proposition~\ref{prop:avar}.
\end{proof}

\begin{corollary}[Variance dominance: OLS plug-in special case]
\label{cor:dominance-ols}
Under the Gaussian linear model the OLS score is
$\psif_0(y_j)\propto(y_j-\bx_j^T\bbeta)$, giving
$\E_{H_\tau}[\psif_0(Y)^2]\ge\epsilon\,\E_{H_\tau}[Y^2]\to\infty$
under increasing-scale contamination.
Theorem~\ref{thm:dominance}(ii)--(iii) therefore apply, confirming
CPP dominance over $\bx_{n+1}^T\widehat\bbeta_n$.
\end{corollary}

\begin{proof}
Under the OLS estimating equation,
$\psif_0(y_j)=(y_j-\bx_j^T\bbeta)/\sigma^2$, so
$\E_{\Gept}[\psif_0(Y)^2]\ge\epsilon\,\E_{H_\tau}[\psif_0(Y)^2]
\ge C\,\epsilon\,\E_{H_\tau}[Y^2]$
for a positive constant $C$.
Since $\E_{H_\tau}[Y^2]\to\infty$ by the increasing-scale assumption,
$V_0(\epsilon,\tau)=\bx_{n+1}^T\Sigma_\beta(\epsilon,\tau)\bx_{n+1}
\to\infty$, verifying the hypothesis of Theorem~\ref{thm:dominance}(ii).
Part~(iii) then gives the dominance conclusion.
\end{proof}

\begin{corollary}[ELPD comparison: Gaussian special case]
\label{cor:elpd-gaussian}
Under $Y\mid a\sim N(a,s^2)$, $\ELPD(a;\Geps)\equiv-(\meps-a)^2/(2s^2)$
up to a constant, so $\lambda=1/(2s^2)$ and
Proposition~\ref{prop:elpd-general} applies with
\begin{align*}
&\ELPD(\astar;\Geps)-\ELPD(\aopt;\Geps)
\\&\approx
\frac{(\mzero-\aopt)^2-(\mzero-\astar)^2}{2s^2}
\\&\quad+
\frac{\epsilon}{s^2}
\Bigl[
(\mzero-\aopt)\{(\mH-\mzero)-\Bzero(H)\}
-(\mzero-\astar)\{(\mH-\mzero)-\Bstar(H)\}
\Bigr].    
\end{align*}
\end{corollary}

\begin{proof}[Proof of Corollary~\ref{cor:elpd-gaussian}]
$\E_{\Geps}[(Y-a)^2]=\mathrm{Var}_{\Geps}(Y)+(\meps-a)^2$, so
$\ELPD(a;\Geps)=-(\mathrm{Var}_{\Geps}(Y)+(\meps-a)^2)/(2s^2)$.
Dropping the variance term (constant in $a$) and substituting
the first-order expansions for $a_{0,\epsilon}$ and $a_{\ast,\epsilon}$
gives the stated expression.
\end{proof}

\subsection{Nonparametric inheritance}
\label{sec:inheritance}

The theoretical results of Sections~\ref{sec:scores}--\ref{sec:contamination}
hold under conditions~(C1) and~(C2) of Section~\ref{sec:general}
and Assumption~\ref{ass:redesc}, with no requirement that $\mu_{1i}(a)$
be affine in $a$. The affine structure $\mu_{1i}(a)=c_i+d_i a$,
which holds for the Gaussian linear model and the nonparametric
models below, additionally yields the rank-two update
(Corollary~\ref{cor:rank2}) and one-dimensional score equation
(Corollary~\ref{cor:1d}), but these are consequences of the affine
special case rather than prerequisites for the general theory.
The following corollary makes the inheritance precise.

\begin{corollary}[Nonparametric inheritance of all Section~2 results]
\label{cor:nonparametric}
Let the data be modelled by either the basis-expansion regression with
feature map $\bx_i\mapsto\bz_i$, or Gaussian process regression with
kernel $k$. In both cases, conditions~(C1) and~(C2) of
Section~\ref{sec:general} are satisfied and $\mu_{1i}(a)=c_i+d_i a$
is affine (verified in Section~\ref{sec:nonparametric}).

\noindent\textit{General results} (require only (C1), (C2), and
Assumption~\ref{ass:redesc}; hold for any model in this class):
\begin{enumerate}
  \item[(i)] Assumption~\ref{ass:redesc} is satisfied and the influence
    function is globally bounded
    (Proposition~\ref{prop:bounded-inf}, Corollary~\ref{cor:scores-gaussian}).
  \item[(ii)] The CPP objective is locally strongly convex near any
    minimizer at which predictive discrepancies are small
    (Lemma~\ref{lem:convexity}, Corollary~\ref{cor:convexity-gaussian}).
  \item[(iii)] The predictor $\anstar$ is $\sqrt{n}$-consistent and
    asymptotically normal (Proposition~\ref{prop:asymp-dist}).
  \item[(iv)] The asymptotic variance dominates any plug-in predictor
    with unbounded influence under $\epsilon$-contamination
    (Theorem~\ref{thm:dominance}).
  \item[(v)] ELPD comparisons under contamination follow
    Proposition~\ref{prop:elpd-general}
    (Corollary~\ref{cor:elpd-gaussian} for Gaussian predictives).
\end{enumerate}

\noindent\textit{Affine-structure results} (additionally require
$\mu_{1i}(a)=c_i+d_i a$; hold for the Gaussian linear model,
basis-expansion, and GP regression):
\begin{enumerate}
  \item[(vi)] The minimizer $\astar$ is found by solving a single
    scalar equation (Corollary~\ref{cor:1d}).
  \item[(vii)] The scalars $c_i$, $d_i$ are computable at $O(np^2)$
    cost via a rank-two Woodbury update
    (Corollary~\ref{cor:rank2}, with $\bx_i$ replaced by $\bz_i$).
\end{enumerate}
\end{corollary}

\subsection{Nonparametric model classes}
\label{sec:nonparametric}

We extend the proposed framework to a nonparametric regression setting,
verifying conditions~(C1) and~(C2) of Section~\ref{sec:general} ---
and hence the affine structure $m_{1i}(a)=c_i+d_i a$ required by
Corollary~\ref{cor:nonparametric} --- for two nonparametric model
classes.

\subsubsection{Basis function route}

\paragraph{Model.}
Let $\{(\bx_i,y_i)\}_{i=1}^n$ with $y_i=f(\bx_i)+\varepsilon_i$,
$\varepsilon_i\sim N(0,\sigma^2)$.
Approximate $f(\bx)=\sum_{k=1}^K\beta_k\phi_k(\bx)$,
and define the transformed covariate
$\bz_i=(\phi_1(\bx_i),\ldots,\phi_K(\bx_i))^T\in\mathbb{R}^K$
with $\bZ\in\mathbb{R}^{n\times K}$ the matrix with rows $\bz_i^T$.
The model reduces to Bayesian linear regression in the transformed features:
\[
\by\mid\bbeta\sim N(\bZ\bbeta,\sigma^2 I_n),
\qquad
\bbeta\sim N(\bbeta_0,\sigma^2\bV).
\]

\paragraph{Posterior quantities.}
Define
\[
\bA=\bZ^T\bZ+\bV^{-1},
\qquad
\bba=\bZ^T\by+\bV^{-1}\bbeta_0,
\qquad
\widehat\bbeta=\bA^{-1}\bba.
\]

\paragraph{Leave-one-out predictive distribution.}
For each $i=1,\ldots,n$, let $\ell_i=\bz_i^T\bA^{-1}\bz_i$. Then
\[
y_i\mid y_{(-i)}\sim N(m_{2i},s_{2i}^2),
\qquad
m_{2i}=\frac{\bz_i^T\widehat\bbeta-\ell_i y_i}{1-\ell_i},
\quad
s_{2i}^2=\frac{\sigma^2}{1-\ell_i}.
\]

\paragraph{Swapped predictive distribution.}
Let $(\bx_{n+1},a)$ be a candidate point and define
$\bz_{n+1}=(\phi_1(\bx_{n+1}),\ldots,\phi_K(\bx_{n+1}))^T$.
For each $i$, define the augmented design
\[
\bZ_{-i}^{(+)}
=
\begin{pmatrix}\bZ_{-i}\\ \bz_{n+1}^T\end{pmatrix},
\qquad
\by_{-i}^{(+)}
=
\begin{pmatrix}\by_{-i}\\ a\end{pmatrix},
\]
and let
\[
\bA_{-i}^{(+)}=(\bZ_{-i}^{(+)})^T\bZ_{-i}^{(+)}+\bV^{-1},
\qquad
\bba_{-i}^{(+)}=(\bZ_{-i}^{(+)})^T\by_{-i}^{(+)}+\bV^{-1}\bbeta_0.
\]
Then
\[
y_i\mid\{(\by_{-i},\bx_{-i}),(a,\bx_{n+1})\}
\sim
N\bigl(m_{1i}(a),s_{1i}^2\bigr),
\]
where
\[
m_{1i}(a)=\bz_i^T(\bA_{-i}^{(+)})^{-1}\bba_{-i}^{(+)},
\qquad
s_{1i}^2=\sigma^2(1+\delta_i),
\qquad
\delta_i=\bz_i^T(\bA_{-i}^{(+)})^{-1}\bz_i.
\]

\paragraph{Affine structure.}
Since $\bba_{-i}^{(+)}$ is affine in $a$,
\[
m_{1i}(a)=c_i+d_i a,
\qquad
c_i=\bz_i^T(\bA_{-i}^{(+)})^{-1}(\bZ_{-i}^T\by_{-i}+\bV^{-1}\bbeta_0),
\quad
d_i=\bz_i^T(\bA_{-i}^{(+)})^{-1}\bz_{n+1}.
\]

\paragraph{Conformity-based projection.}
We define
\[
a^\ast
=
\arg\min_a
\sum_{i=1}^n
D\!\left(
N(m_{1i}(a),s_{1i}^2),\,N(m_{2i},s_{2i}^2)
\right).
\]

The entire framework is thus identical to the parametric case after replacing
$\bx_i$ with $\bz_i$. The affine representation $m_{1i}(a)=c_i+d_i a$ is
preserved, ensuring tractable optimization over $a$.

\subsubsection{Gaussian process regression}

We now extend the proposed framework to a fully nonparametric setting
using Gaussian process priors \citep{Rasmussen2006}.

\paragraph{Model.}
Let $\{(\bx_i,y_i)\}_{i=1}^n$ with $y_i=f(\bx_i)+\varepsilon_i$,
$\varepsilon_i\sim N(0,\sigma^2)$, and $f\sim GP(m,k)$.
Define
\[
\bK_n=\bigl[k(\bx_i,\bx_j)\bigr]_{i,j=1}^n,
\qquad
\bSigma_n=\bK_n+\sigma^2 I_n.
\]

\paragraph{Posterior predictive distribution.}
For any $\bx$, define
$\bk_n(\bx)=(k(\bx,\bx_1),\ldots,k(\bx,\bx_n))^T$ and
$\bmm_n=(m(\bx_1),\ldots,m(\bx_n))^T$.
Then
\[
y(\bx)\mid\mathcal{D}_n\sim N\!\bigl(m_n(\bx),\,v_n(\bx)+\sigma^2\bigr),
\]
where
\[
m_n(\bx)=m(\bx)+\bk_n(\bx)^T\bSigma_n^{-1}(\by-\bmm_n),
\qquad
v_n(\bx)=k(\bx,\bx)-\bk_n(\bx)^T\bSigma_n^{-1}\bk_n(\bx).
\]

\paragraph{Leave-one-out predictive distribution.}
For each $i$, analogous expressions hold with the $i$th point removed:
\[
y_i\mid y_{(-i)}\sim N(m_{2i},s_{2i}^2),
\]
with $m_{2i}$ and $s_{2i}^2$ given by the standard GP leave-one-out
formulas.

\paragraph{Swapped predictive distribution.}
Let $(\bx_{n+1},a)$ be a candidate point. Then
\[
y_i\mid\mathcal{D}_{-i}^{(+)}\sim N\!\bigl(m_{1i}(a),s_{1i}^2\bigr).
\]

\paragraph{Affine structure in $a$.}
Since the GP posterior mean is linear in the observations $\by$, and
the augmented observation vector $\by_{-i}^{(+)}$ is affine in $a$,
the posterior mean $m_{1i}(a)$ inherits this affine dependence:
\[
m_{1i}(a)=c_i+d_i a,
\]
where $d_i$ corresponds to the influence of the added point
$(\bx_{n+1},a)$ in the augmented system.

\paragraph{Conformity-based projection.}
We define
\[
a^\ast
=
\arg\min_a
\sum_{i=1}^n
D\!\left(
N(m_{1i}(a),s_{1i}^2),\,N(m_{2i},s_{2i}^2)
\right).
\]

The Gaussian predictive structure is preserved under GP regression,
ensuring that all divergence calculations remain explicit. The affine
dependence of $m_{1i}(a)$ on $a$ guarantees tractable optimization,
providing a fully nonparametric Bayesian extension of the proposed
framework.

\section{Numerical illustration under linear model}
\subsection{Implementation under Unknown Variance}
\label{sec:unknown-variance}

The general CPP criterion~\eqref{eq:cpp} requires the predictive
distributions $\pii$ and $\qia$ to be available in closed form
(condition~(C1)). Under the Gaussian linear model of
Section~\ref{sec:loo}, these distributions depend on the noise
variance $\sigma^2$, which in practice is unknown. This section
addresses implementation in that specific setting: we show how
posterior uncertainty about $\sigma^2$ can be propagated through
the projection procedure while preserving the theoretical properties
established in Sections~\ref{sec:scores}--\ref{sec:contamination}.
The two approaches developed here are specific to the Gaussian
linear model but illustrate a general strategy --- integrating the
CPP objective over a posterior for the nuisance parameter --- that
extends to other parametric settings satisfying (C1) and (C2).

We adopt a Bayesian plug-in strategy: obtain posterior
draws $\{\sigma^{2,(t)}\}_{t=1}^T\sim\Pi(\sigma^2\mid (x_i,y_i)_{i=1}^n)$ and
propagate uncertainty through the projection procedure.

\subsubsection{Approach I: Draw-wise plug-in optimization}

For each draw $\sigma^{2,(t)}$, solve the projection problem to obtain
$a^{*(t)}=T(\sigma^{2,(t)})$. The sample $\{a^{*(t)}\}$ is then
summarized by its mean or median.

Under the $\epsilon$-contamination model $\Geps=(1-\epsilon)F_0+\epsilon H$, where $g_\epsilon$ is the corresponding density and the posterior concentrates around the pseudo-true value
$\sigma_\epsilon^2=\arg\min_{\sigma^2}\mathrm{KL}(g_\epsilon\|f_\sigma)$, where $f_\sigma$ is the parametric density of the working model for $\sigma^2$.

\begin{theorem}[Contamination expansion of pseudo-true variance]
\label{thm:pseudotrue}
Under regularity conditions,
\[
\sigma_\epsilon^2
=
\sigma_0^2+\epsilon B_\sigma(H)+o(\epsilon),
\]
where $\sigma_0^2=\arg\min\mathrm{KL}(f_0\|f_\sigma)$ and
$B_\sigma(H)=-\Psi_H(\sigma_0^2)/\dot\Psi_{F_0}(\sigma_0^2)$.
\end{theorem}

\begin{proof}
Let $\Psi_A(\sigma^2)=\E_A[\partial_{\sigma^2}\log f_\sigma(Y)]$ for a generic distribution $A$.
The pseudo-true value satisfies
$\Psi(\sigma_\epsilon^2,\epsilon):=(1-\epsilon)\Psi_{F_0}(\sigma_\epsilon^2)
+\epsilon\Psi_H(\sigma_\epsilon^2)=0$.
At $\epsilon=0$, $\Psi_{F_0}(\sigma_0^2)=0$ and
$\partial_{\sigma^2}\Psi_{F_0}(\sigma_0^2)\neq 0$ by second-order
identifiability. The implicit function theorem gives
$d\sigma_\epsilon^2/d\epsilon|_{\epsilon=0}
=-\Psi_H(\sigma_0^2)/\dot\Psi_{F_0}(\sigma_0^2)=B_\sigma(H)$.
\end{proof}

\paragraph{Implication for the projection target.}
Let $a^\ast(\sigma^2)=T(\sigma^2)$. Then
\[
a_\epsilon^\ast
=
a_0^\ast+\epsilon T'(\sigma_0^2) B_\sigma(H)+o(\epsilon).
\]
Thus the plug-in rule inherits a first-order contamination bias.

\paragraph{Sample-level behavior.}
Under posterior concentration at $\sigma_\epsilon^2$ and continuity of $T$,
\[
a^{*(t)}=T(\sigma^{2,(t)})=T(\sigma_\epsilon^2)+o_p(1).
\]
Since $\sigma_\epsilon^2=\sigma_0^2+\epsilon B_\sigma(H)+o(\epsilon)$,
a first-order Taylor expansion gives
\[
a^{*(t)}
=
a_0^*+\epsilon T'(\sigma_0^2)B_\sigma(H)+o(\epsilon)+o_p(1).
\]
Hence the procedure reflects both contamination bias $o(\epsilon)$ and
posterior uncertainty $o_p(1)$.

\subsubsection{Approach II: Posterior-averaged optimization}

Define $\hat{a}_B=\arg\min_{a\in\mathcal{A}}B^{-1}\sum_{t=1}^B J(a;\sigma^{2,(t)})$ based on $B$ posterior samples
with population counterpart $a_n^*=\arg\min_a Q_n(a)$,
$Q_n(a)=\int J(a;\sigma^2)\,d\Pi_n(\sigma^2)$.

\begin{proposition}[Convergence of averaged objective]
\label{prop:avgopt}
Assume $\mathcal{A}$ is compact, $J(a;\sigma^2)$ is continuous,
$\Pi_n$ concentrates at $\sigma_\epsilon^2$, and the minimizer of
$J(a;\sigma_\epsilon^2)$ over $\mathcal{A}$ is unique. Then
$a_n^*\to a_\epsilon^*:=\arg\min_a J(a;\sigma_\epsilon^2)$, and with
$\sigma_\epsilon^2=\sigma_0^2+\epsilon B_\sigma(H)+o(\epsilon)$
(Theorem~\ref{thm:pseudotrue}),
\[
a_\epsilon^*
=
a_0^*+\epsilon B_a(H)+o(\epsilon),
\qquad
B_a(H)=-H_a(\sigma_0^2)^{-1}\nabla_{\sigma^2}\nabla_a J(a_0^*;\sigma_0^2)\,B_\sigma(H),
\]
where $H_a(\sigma_0^2):=\partial^2_a J(a_0^*;\sigma_0^2)>0$ by
Lemma~\ref{lem:convexity}.
\end{proposition}

Both approaches are asymptotically equivalent at first order:
$a_n^*=a_0^*+\epsilon B_a(H)+o(\epsilon)$.
Approach~I produces a distribution of solutions; Approach~II minimizes
a smoothed objective and is typically more stable numerically.
In both cases, contamination affects the procedure through the shift in
the pseudo-true variance parameter, a mechanism that is specific to the
Gaussian working model but reflects a general phenomenon: uncertainty
in any nuisance parameter of the predictive distributions propagates
into the CPP criterion through condition~(C1).

To mitigate variance inflation under contamination, one may replace
$\sigma^{2,(t)}$ by a transformed version $R(\sigma^{2,(t)})$
(e.g., truncation or robust scaling). This effectively attenuates the
contamination bias by reducing the sensitivity of the objective to
extreme posterior draws. In other model classes satisfying (C1) and
(C2), the analogous strategy is to integrate the CPP objective
$J(a;\theta)$ over a posterior for the relevant nuisance parameter
$\theta$, with the contamination bias governed by the influence
function of the posterior for $\theta$ under the working model.
\subsection{Simulation experiments}
\label{sec:simulations}

The general CPP framework of Section~\ref{sec:model} applies to any
model satisfying conditions~(C1) and~(C2). The simulations here
illustrate its finite-sample behavior in the Gaussian linear model
setting of Section~\ref{sec:loo}, which provides the closed-form
predictive distributions and affine structure that make the experiments
fully tractable. The comparator throughout is the posterior mean
(MAP-style plug-in predictor), which corresponds to the OLS plug-in
$\hat{f}_n(\bx_{n+1})=\bx_{n+1}^T\widehat\bbeta_n$ discussed in
Corollaries~\ref{cor:avar-ols} and~\ref{cor:dominance-ols}.

The primary performance metric is the mean log-predictive density
difference
\[
\mathrm{MLPD}
=
\frac{1}{n_{\mathrm{test}}}\sum_{i}
\bigl\{\log p_{\mathrm{CPP}}(y_i\mid x_i)-\log p_{\mathrm{MAP}}(y_i\mid x_i)\bigr\},
\]
so that positive values indicate superior predictive density for CPP.
We apply Approach~I from Section~\ref{sec:unknown-variance} for
variance propagation. Results used $50$ Monte Carlo replicates, $500$
posterior draws for $\sigma^2$, and a grid of length $61$ for
optimization over $a$. The prior was held fixed with a diffuse Gaussian
on $\bbeta$ and inverse-gamma hyperparameters $a_0=b_0=0.1$.

For each replicate, data were generated from $y=X\beta+\varepsilon$
with $X\in\mathbb{R}^{n\times p}$ having i.i.d.\ standard normal
entries, true coefficients $(1,-1,0.5,0,\ldots,0)^T$, and
$\varepsilon\sim N(0,\sigma^2 I_n)$. Response contamination was
introduced by adding large Gaussian perturbations to a randomly
selected fraction of the responses, mimicking the $\epsilon$-contamination
model $G_\epsilon=(1-\epsilon)F_0+\epsilon H$ of Section~\ref{sec:asymptotics}.
After contamination, each covariate column was standardized.

\subsubsection{Effect of response outliers}

Fixing $(n,p,\sigma)=(200,6,1)$ and varying the response-outlier
fraction from $0.03$ to $0.05$, CPP showed a clear and systematic advantage: the mean MLPD
was $0.2535$ at $3\%$, $0.1816$ at $4\%$, $0.1568$ at $5\%$, with CPP outperforming MAP
in every replicate. The largest gains appeared at low-to-moderate
outlier levels, consistent with the redescending nature of the
divergence scores.

\begin{table}[t]
\centering
\caption{Summary of Experiments 1. Positive MLPD values favor CPP over MAP.}
\label{tab:sim_outlier_leverage}
\small
\begin{tabular}{llrrrrr}
\hline
Experiment & Setting & $\widehat{\mathrm{MLPD}}$ & SE & 95\% CI lower & 95\% CI upper & \% positive \\
\hline
Outlier fraction & 0.03 & 0.2535 & 0.0466 & 0.1622 & 0.3449 & 100.00 \\
Outlier fraction & 0.04 & 0.1816 & 0.0269 & 0.1289 & 0.2342 & 100.00 \\
Outlier fraction & 0.05 & 0.1568 & 0.0304 & 0.0972 & 0.2165 & 100.00 \\
\hline
\end{tabular}
\end{table}

\subsubsection{Sensitivity to the DPD tuning parameter}

Fixing the same contamination regime and varying
$\alpha\in\{0.1,0.3,0.5,0.75,1,1.5,2\}$, CPP consistently outperformed
MAP for all values. The mean MLPD ranged from about $0.0353$ to
$0.0463$, with the largest improvement near $\alpha=1.0$ and stable
performance across the rest of the range.

\begin{table}[t]
\centering
\caption{Sensitivity to DPD tuning parameter.}
\label{tab:sim_divergence_dpd}
\small
\begin{tabular}{llrrrrr}
\hline
Experiment & Setting & $\widehat{\mathrm{MLPD}}$ & SE & 95\% CI lower & 95\% CI upper & \% positive \\
\hline
DPD tuning & $\alpha=0.10$ & $0.0358$ & $0.0025$ & $0.0310$ & $0.0406$ & $100$ \\
DPD tuning & $\alpha=0.30$ & $0.0413$ & $0.0037$ & $0.0341$ & $0.0484$ & $100$ \\
DPD tuning & $\alpha=0.50$ & $0.0410$ & $0.0031$ & $0.0349$ & $0.0471$ & $100$ \\
DPD tuning & $\alpha=0.75$ & $0.0440$ & $0.0032$ & $0.0377$ & $0.0502$ & $100$ \\
DPD tuning & $\alpha=1.00$ & $0.0463$ & $0.0046$ & $0.0372$ & $0.0554$ & $100$ \\
DPD tuning & $\alpha=1.50$ & $0.0353$ & $0.0022$ & $0.0310$ & $0.0397$ & $100$ \\
DPD tuning & $\alpha=2.00$ & $0.0449$ & $0.0034$ & $0.0382$ & $0.0515$ & $100$ \\
\hline
\end{tabular}
\end{table}

\subsubsection{Scaling with sample size}

Fixing $p=4$, $\sigma=1$, and $10\%$ outlier fraction while varying
$n\in\{50,100,200,400\}$, CPP improved upon MAP uniformly. The mean
MLPD was $0.7595$ for $n=50$, $0.2126$ for $n=100$, $0.0734$ for
$n=200$, and $0.0288$ for $n=400$. The decreasing magnitude is
consistent with individual outliers having diminishing influence on the
posterior as sample size grows.

\subsubsection{Scaling with predictor dimension}

Fixing $n=200$, $\sigma=1$, and $10\%$ outlier fraction while varying
$p\in\{2,4,6,10,15\}$, CPP dominated MAP throughout, with mean MLPD
values between $0.061$ and $0.082$. Gains were positive in every
replicate for every $p$.

\subsubsection{Effect of noise level}

Fixing $(n,p)=(200,6)$ with $10\%$ outlier fraction and varying
$\sigma\in\{0.5,1,2,4\}$, CPP again uniformly improved on MAP, with
mean MLPD values approximately $0.0753$, $0.0597$, $0.0685$, and
$0.0609$.

\begin{table}[t]
\centering
\caption{Summary of scaling and noise experiments.}
\label{tab:sim_scaling_noise}
\small
\begin{tabular}{llrrrrr}
\hline
Experiment & Setting & $\widehat{\mathrm{MLPD}}$ & SE & 95\% CI lower & 95\% CI upper & \% positive \\
\hline
Sample size & $n=50$  & $0.7595$ & $0.1931$ & $0.3811$ & $1.1379$ & $100$ \\
Sample size & $n=100$ & $0.2126$ & $0.0393$ & $0.1355$ & $0.2896$ & $100$ \\
Sample size & $n=200$ & $0.0734$ & $0.0100$ & $0.0539$ & $0.0930$ & $100$ \\
Sample size & $n=400$ & $0.0288$ & $0.0015$ & $0.0259$ & $0.0316$ & $100$ \\
\hline
Dimension & $p=2$  & $0.0610$ & $0.0041$ & $0.0529$ & $0.0691$ & $100$ \\
Dimension & $p=4$  & $0.0615$ & $0.0039$ & $0.0539$ & $0.0691$ & $100$ \\
Dimension & $p=6$  & $0.0823$ & $0.0171$ & $0.0487$ & $0.1159$ & $100$ \\
\hline
Noise level & $\sigma=0.5$ & $0.0753$ & $0.0101$ & $0.0555$ & $0.0951$ & $100$ \\
Noise level & $\sigma=1$   & $0.0597$ & $0.0046$ & $0.0507$ & $0.0688$ & $100$ \\
Noise level & $\sigma=2$   & $0.0685$ & $0.0082$ & $0.0524$ & $0.0846$ & $100$ \\
Noise level & $\sigma=4$   & $0.0609$ & $0.0059$ & $0.0492$ & $0.0725$ & $100$ \\
\hline
\end{tabular}
\end{table}

\subsubsection{Clean-data baseline}

Under correctly specified clean data with no contamination, the gains
were intentionally small: mean MLPD approximately $0.00067$ for DPD and
$0.00103$ for Hellinger. These near-zero differences confirm that CPP
does not materially degrade performance when contamination is absent.

\begin{table}[t]
\centering
\caption{Performance under clean data.}
\label{tab:sim_clean}
\small
\begin{tabular}{lrrrrr}
\hline
Divergence & $\widehat{\mathrm{MLPD}}$ & SE & 95\% CI lower & 95\% CI upper & \% positive \\
\hline
DPD      & $0.0007$ & $0.0001$ & $0.0005$ & $0.0009$ & $76$ \\
Hellinger & $0.0010$ & $0.0002$ & $0.0007$ & $0.0013$ & $80$ \\
\hline
\end{tabular}
\end{table}

\subsubsection{Summary}

These experiments illustrate the general theoretical results of
Section~\ref{sec:model} in the Gaussian linear model setting.
Three conclusions emerge consistently. First, CPP provides
non-negligible predictive-density gains under response contamination,
most pronounced at low-to-moderate outlier levels ($5$--$10\%$),
consistent with the asymptotic dominance of Theorem~\ref{thm:dominance}
and Corollary~\ref{cor:dominance-ols}. Second, the advantage is stable
across DPD tuning parameters, sample sizes, predictor dimensions, and
noise levels, reflecting the generality of the bounded-influence
property established in Proposition~\ref{prop:bounded-inf} and
Corollary~\ref{cor:scores-gaussian}. Third, under clean data CPP
performs essentially identically to the Bayesian plug-in predictor,
confirming the efficiency equivalence of Corollary~\ref{cor:clean}
and showing that the robustness mechanism carries no meaningful
efficiency cost under ideal conditions.

\subsection{Data Applications}
\label{sec:data}

The two analyses here illustrate the
Gaussian linear model instantiation of Section~\ref{sec:loo} on
real datasets with documented response outliers, serving as
empirical counterparts to the asymptotic theory of
Sections~\ref{sec:asymptotics}--\ref{sec:contamination} with unknown
variance propagated via Approach~I of
Section~\ref{sec:unknown-variance}. 

In each case, we adopt a repeated random-split design: all outlying
observations are held out in every test set together with a random
draw of clean observations, and the split is repeated 10 times.
The primary metric is the mean log-predictive density difference
\[
\mathrm{MLPD}
=
n_{\mathrm{test}}^{-1}
\sum_{i\in\mathrm{test}}
\bigl\{
\log p_{\mathrm{CPP}}(y_i\mid\bx_i)
-\log p_{\mathrm{MAP}}(y_i\mid\bx_i)
\bigr\},
\]
with positive values indicating superior predictive density for CPP.
In both applications the prior is $\bbeta\sim N(\mathbf{0},100\,I_p)$
with inverse-gamma hyperparameters $a_0=b_0=0.1$, variance
uncertainty is propagated via $B=500$ posterior draws, and the CPP
grid has length 61 centered on the MAP prediction with a
$\pm4\hat\sigma$ window.

\subsubsection{Body fat prediction}
\label{sec:data-bodyfat}

The body fat dataset of \citet{Koenig1994} from the
\texttt{TH.data} R package \citep{Hothorn2014}, contains $n=71$
observations for women. The response \texttt{DEXfat} is the
percentage of body fat measured by dual-energy X-ray absorptiometry
(DXA), and the nine predictors are age, waist and hip
circumferences, elbow and knee breadths, and four composite
anthropometric indices. DXA measurements are subject to instrument
drift and patient-positioning errors that produce genuine
gross-error outliers in the response. Inspection of the OLS fit
reveals three observations with studentised residuals exceeding
$2.5$ in absolute value (observations 27, 41, and 48, with
\texttt{DEXfat} values of $40.6\%$, $60.7\%$, and $62.0\%$,
respectively), constituting $4.2\%$ of the sample. The value
$62.0\%$ is physiologically implausible for the corresponding
anthropometric profile and is likely a recording or calibration
error. All variables are standardized before analysis.

Each of the 10 splits holds out the three outlier observations
together with 17 randomly selected clean observations
($n_{\mathrm{test}}=20$, $n_{\mathrm{train}}=51$). The per-split
MLPD for CPP-DPD ranges from $+0.217$ to $+0.267$ with a mean of
$+0.241$ (SE $0.005$); CPP-Hellinger yields a mean of
$+0.227$ (SE $0.005$). Both criteria outperform MAP in all 10
splits.

Stratifying by observation type confirms that the gain is entirely
attributable to the contaminated observations. On the 68 distinct
clean observations the average MLPD is $-0.001$, negligible across
all splits. On the three outlier observations the average gain for
CPP-DPD is $+1.612$ nats, monotonically increasing with outlier
severity: $+0.943$ at $40.6\%$, $+1.192$ at $60.7\%$, and $+2.700$
at $62.0\%$.

\subsubsection{New York air quality}
\label{sec:data-airquality}

The New York Air Quality dataset \citep{Chambers1983} records daily
atmospheric measurements taken from May through September 1973.
After removing the 42 observations with missing values, the working
sample consists of $n=111$ complete cases. The response is daily
ozone concentration (ppb) and the four predictors are solar
radiation, wind speed, maximum temperature, and calendar month.
Two observations exceed the $1.5\times\mathrm{IQR}$ upper fence
($135$ and $168$\,ppb, or $1.8\%$ of the sample), representing
sensor spikes or extreme photochemical events. All variables are
standardized before analysis.

Each of the 10 splits holds out both outlier observations together
with 18 randomly selected clean observations ($n_{\mathrm{test}}=20$,
$n_{\mathrm{train}}=91$). The per-split MLPD for CPP-DPD ranges
from $+0.136$ to $+0.183$ with a mean of $+0.150$ (SE $0.005$);
CPP-Hellinger yields a mean of $+0.141$ (SE $0.004$). Both
criteria outperform MAP in all 10 splits. On the 89 distinct clean
observations the average MLPD is $+0.005$, essentially zero. On
the two outlier observations the CPP-DPD gain is $+0.388$ at
$135$\,ppb and $+2.521$ at $168$\,ppb, with the larger gain at
the more extreme value.

\subsubsection{Summary}

Table~\ref{tab:data_summary} collects the results across both
applications. These real-data illustrations confirm the theoretical
predictions of Section~\ref{sec:model} in the Gaussian linear model
setting. Three findings emerge consistently. First, CPP dominates
MAP in every one of the 10 splits in both datasets, in agreement
with Theorem~\ref{thm:dominance} and Corollary~\ref{cor:dominance-ols}.
Second, the gains on clean observations are negligible in both
datasets (average MLPD within $\pm 0.005$ nats), consistent with
Theorem~\ref{thm:consistency} and Corollary~\ref{cor:clean}, and the
clean-data baseline of Section~\ref{sec:simulations}. Third, the
gains on the outlier observations are large, positive, and increasing
in outlier severity, directly reflecting the redescending influence
established in Proposition~\ref{prop:bounded-inf} and
Corollary~\ref{cor:scores-gaussian}, and the asymptotic dominance
under contamination of Theorem~\ref{thm:dominance}.

\begin{table}[t]
\centering
\caption{Summary of repeated random-split results (10 splits each).
MLPD is the mean log-predictive density gain (CPP$-$MAP) averaged
over splits; standard errors in parentheses. Gain columns report
the mean gain on clean and outlier test observations pooled across
all splits.}
\label{tab:data_summary}
\small
\begin{tabular}{llrrrr}
\hline
Dataset & Method
  & MLPD (SE)
  & Splits $>0$
  & Gain (clean)
  & Gain (outlier) \\
\hline
\multirow{2}{*}{Body fat}
  & CPP-DPD ($\alpha=1$)
    & $+0.241\;(0.005)$ & $10/10$ & $-0.001$ & $+1.612$ \\
  & CPP-Hellinger
    & $+0.227\;(0.005)$ & $10/10$ & $-0.002$ & $+1.525$ \\[4pt]
\multirow{2}{*}{Air quality}
  & CPP-DPD ($\alpha=1$)
    & $+0.150\;(0.005)$ & $10/10$ & $+0.005$ & $+1.455$ \\
  & CPP-Hellinger
    & $+0.141\;(0.004)$ & $10/10$ & $+0.005$ & $+1.368$ \\
\hline
\end{tabular}
\end{table}

Figure~\ref{fig:data_splits} displays the split-level MLPD and
per-observation gain distributions for each dataset.
Panel~(a) shows that all 10 per-split MLPDs are positive and
well-separated from zero in both applications, confirming the
stability of the result across partitions. Panel~(b) shows a sharp
separation between the clean-observation gains (centered near zero)
and the outlier gains (large and positive), satisfying the theoretical prediction that the bounded influence of the CPP score
concentrates gains precisely at contaminated observations.

\begin{figure}[t]
\centering
\includegraphics[width=\linewidth]{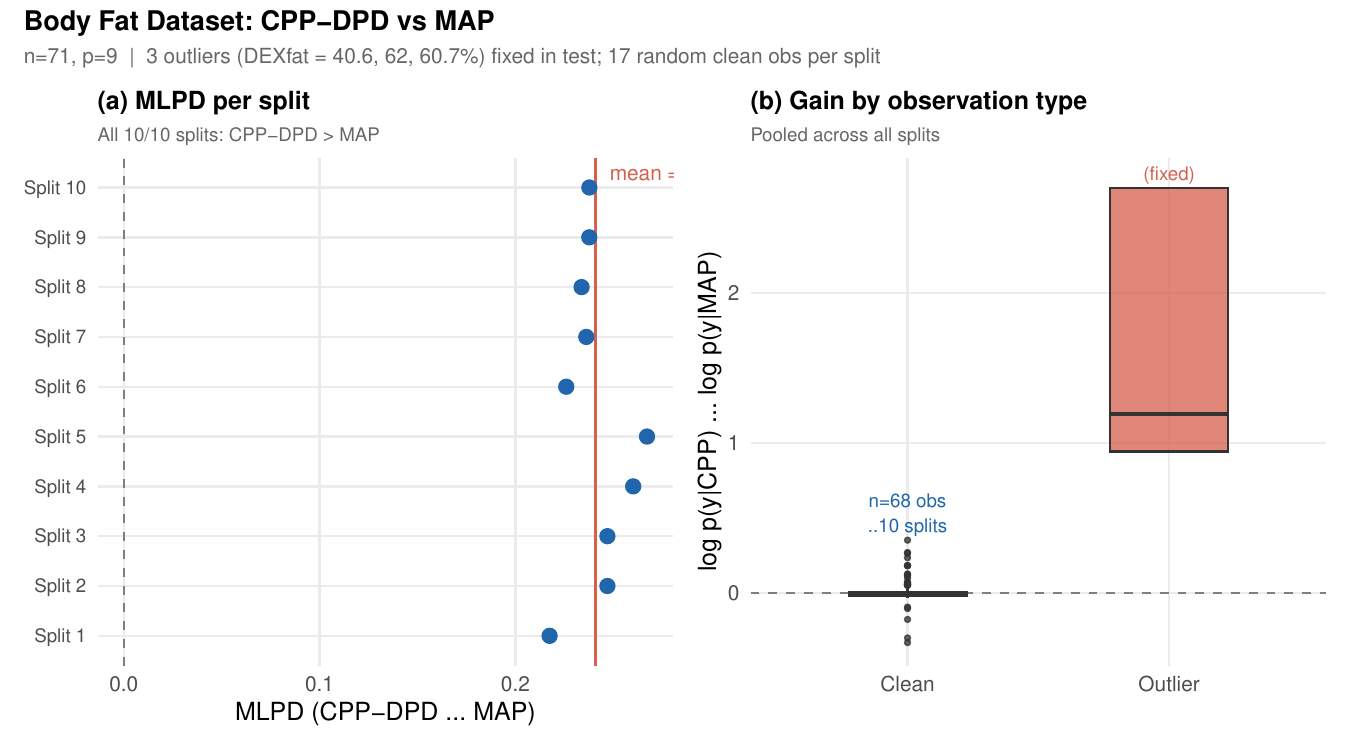}
\vspace{6pt}
\includegraphics[width=\linewidth]{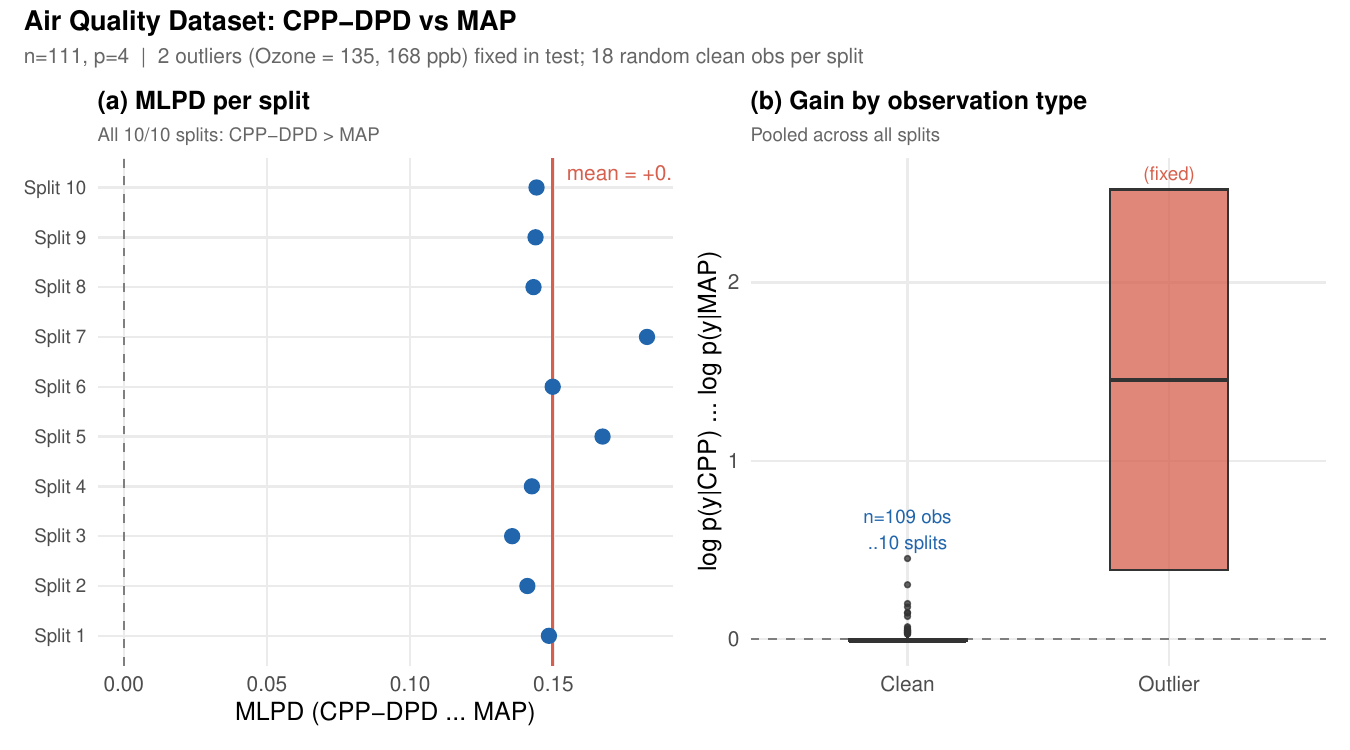}
\caption{Repeated random-split results (10 splits) for the body fat
dataset (top) and the air quality dataset (bottom). Panel~(a): per-split
MLPD for CPP-DPD; red vertical line is the mean. Panel~(b): boxplots
of the per-observation log-predictive density gain (CPP-DPD$-$MAP),
stratified by clean vs outlier test observations and pooled across
all splits. In both datasets CPP-DPD outperforms MAP on every split,
outlier gains are large and positive, and clean-observation gains
are negligible.}
\label{fig:data_splits}
\end{figure}

\section{Discussion}
\label{sec:discussion}
The general CPP criterion~\eqref{eq:cpp} and the associated theory
of Sections~\ref{sec:scores}--\ref{sec:contamination} require only
conditions~(C1) and~(C2): closed-form predictive distributions and
a differentiable swapped mean $\mu_{1i}(a)$. The affine special case
$\mu_{1i}(a)=c_i+d_i a$, which holds for the three model classes
verified in Sections~\ref{sec:loo} and~\ref{sec:nonparametric}
(Corollary~\ref{cor:nonparametric}), additionally yields the
rank-two computational update and one-dimensional score equation.
Several directions for future work follow naturally from this
structure.

For generalized linear models (GLMs), conjugacy is unavailable and
the posterior mean is no longer linear in $\by$, so the affine
special case does not hold. The rank-two Woodbury update
(Corollary~\ref{cor:rank2}) and the closed-form one-dimensional
solution (Corollary~\ref{cor:1d}) therefore do not apply directly.
We plan to address the GLM setting in a companion paper, where the
central strategy will be to satisfy conditions~(C1) and~(C2)
approximately via Laplace or other second-order approximations to the
posterior, under which the swapped predictive mean $\mu_{1i}(a)$
remains a differentiable function of $a$. Under such approximations,
the general bounded-influence result (Proposition~\ref{prop:bounded-inf}),
the local convexity lemma (Lemma~\ref{lem:convexity}), and the
asymptotic dominance theorem (Theorem~\ref{thm:dominance}) all carry
over, since none of these require the affine structure. The minimizer
$\astar$ will be found by numerical optimization over $a$, and the
LOO approximations of \citet{Vehtari2017} can avoid refitting the
model for each candidate value. Establishing formal theoretical
guarantees and robustness properties under these approximations will
be the primary focus of that work.

The choice of divergence $D$ in~\eqref{eq:cpp} determines the
robustness--efficiency trade-off: the DPD parameter $\alpha$ controls
this balance, with larger $\alpha$ giving stronger robustness at some
efficiency cost under $F_0$. The simulation results of
Section~\ref{sec:simulations} suggest $\alpha=1$ performs well
across a range of contamination levels, but a data-adaptive selection
rule --- for instance, one that minimizes an estimated predictive risk
or exploits the general ELPD framework of
Proposition~\ref{prop:elpd-general} --- would strengthen the
practical applicability of the method.

The general ELPD result (Proposition~\ref{prop:elpd-general}) applies
to any strictly proper scoring rule $S$ for which $S(a;G)\equiv
-\lambda(\meps-a)^2$ locally. Extending the CPP framework to
Student-$t$, Laplace, or skew-normal predictive distributions, which
arise naturally under heavy-tailed or asymmetric error models, would
broaden the scope of conditions~(C1) and~(C2) while preserving the
general theoretical structure of
Sections~\ref{sec:scores}--\ref{sec:contamination}.

\section*{Declaration of the use of generative AI and AI-assisted
technologies}

During the preparation of this work, the authors used Claude and
ChatGPT to draft and aid with the simulation and the analysis. After using
these tools, the authors reviewed and edited the content as necessary
and take full responsibility for the content of the publication.

\bibliographystyle{plainnat}
\bibliography{sample}

\appendix
\section{Proofs}
\label{app:proofs}

\subsection{Proof of Proposition~\ref{prop:asymp-dist}}
\label{app:proof-asymp-dist}

We follow the argument of \citet{Foutz1977}, applying the Inverse
Function Theorem to establish existence, consistency, and uniqueness of
$\widehat{a}_n$ before deriving the asymptotic distribution.

Let $\mathcal{I}(a)=\E[-\partial_a\psif(Y;a)]$.
The domination condition in~(A2) and dominated convergence imply
$\mathcal{I}(a)$ is continuous near $\aopt$, with
$\mathcal{I}(\aopt)=|\Apsi|>0$ by~(A1).
The same domination condition gives the uniform law of large numbers
\[
\sup_{a\in U_\delta}|\PsiEmp'(a)+\mathcal{I}(a)|\pto 0.
\]

Let $\lambda=\tfrac{1}{4}\mathcal{I}(\aopt)>0$.
By the WLLN applied to $-\partial_a\psif(Y;\aopt)$,
$\lambda_n:=\tfrac{1}{4}|\PsiEmp'(\aopt)|\pto\lambda$.
Choose $\delta>0$ small enough that on
$U_\delta=[\aopt-\delta,\aopt+\delta]$:
$|\mathcal{I}(a)-\mathcal{I}(\aopt)|<\tfrac{1}{2}\lambda$ by
continuity of $\mathcal{I}$, and
$\sup_{a\in U_\delta}|\PsiEmp'(a)+\mathcal{I}(a)|>\tfrac{1}{4}\lambda$
with probability tending to zero.
For $a\in U_\delta$, with probability tending to $1$:
\begin{align*}
|\PsiEmp'(a)-\PsiEmp'(\aopt)|
&\le
|\PsiEmp'(a)+\mathcal{I}(a)|
+|\mathcal{I}(a)-\mathcal{I}(\aopt)|
+|\PsiEmp'(\aopt)+\mathcal{I}(\aopt)|\\
&\le
\tfrac{1}{4}\lambda+\tfrac{1}{2}\lambda+\tfrac{1}{4}\lambda=\lambda.
\end{align*}

By the Inverse Function Theorem, with probability tending to $1$,
$\PsiEmp$ is $1$-to-$1$ on $U_\delta$ and its image contains the
interval around $\PsiEmp(\aopt)$ with half-length $\lambda_n\delta$.
Since $\E[\psif(Y;\aopt)]=0$ (as $\aopt$ solves $\PsiPop(a)=0$),
the WLLN gives $|\PsiEmp(\aopt)|<\tfrac{1}{2}\lambda_n\delta$ with
probability tending to $1$, so $0\in\PsiEmp(U_\delta)$ and the root
$\widehat{a}_n=\PsiEmp^{-1}(0)\in U_\delta$ exists.
Since $\delta$ is arbitrary, $\widehat{a}_n\pto\aopt$.
Uniqueness in the sense of \citet{Huzurbazar1948} follows from the
$1$-to-$1$ property.

For the asymptotic distribution, the mean-value expansion
$0=\PsiEmp(\aopt)+(\widehat{a}_n-\aopt)\PsiEmp'(\tilde{a}_n)$
gives
\[
\sqrt{n}(\widehat{a}_n-\aopt)
=
-\frac{\sqrt{n}\,\PsiEmp(\aopt)}{\PsiEmp'(\tilde{a}_n)}.
\]
By~(A3) and the CLT, $\sqrt{n}\,\PsiEmp(\aopt)\dto N(0,\Bpsi)$.
Since $\widehat{a}_n\pto\aopt$, we have $\tilde{a}_n\pto\aopt$, and
the uniform convergence established above gives
$\PsiEmp'(\tilde{a}_n)\pto-\mathcal{I}(\aopt)=\Apsi$.
Slutsky's theorem gives
$\sqrt{n}(\widehat{a}_n-\aopt)\dto N(0,\Bpsi/\Apsi^2)$.
\qed

\subsection{Proof of Proposition~\ref{prop:avar}}
\label{app:proof-avar}

The first claim $\avar(\anstar)=\Bpsi/\Apsi^2$ follows directly from
Proposition~\ref{prop:asymp-dist}. The second claim
$\avar(\hat{f}_n(\bx_{n+1}))=V_0$ holds by assumption on the
plug-in predictor. For the Gaussian linear model with OLS plug-in,
$V_0=\bx_{n+1}^T\Sigma_\beta\bx_{n+1}$ by the delta method applied
to $g(\bbeta)=\bx_{n+1}^T\bbeta$ (Corollary~\ref{cor:avar-ols}). \qed

\subsection{Proof of Theorem~\ref{thm:dominance}}
\label{app:proof-dominance}

\paragraph{Part (i).}
Since $|\psif^\ast(y;a)|\le M$ (Proposition~\ref{prop:bounded-inf}),
$B^\ast(\epsilon,\tau)=\Var_{\Gept}(\psif^\ast(Y;a_{\epsilon,\tau}^\ast))
\le M^2$.
By Proposition~\ref{prop:finite-range},
$|A^\ast(\epsilon,\tau)|\ge\underline{A}_T>0$ for $0<\tau\le T$. Hence
$V^\ast(\epsilon,\tau)\le M^2/\underline{A}_T^2<\infty$.

\paragraph{Part (ii).}
By assumption $\E_{H_\tau}[\psif_0(Y)^2]\to\infty$ as $\tau\to\infty$,
forcing $V_0(\epsilon,\tau)\to\infty$. Under the Gaussian linear model
this is verified by Corollary~\ref{cor:dominance-ols}.

\paragraph{Part (iii).}
Let $C^\ast=\sup_\tau V^\ast(\epsilon,\tau)<\infty$.
By~(ii), there exists $\tau_0$ with $V_0(\epsilon,\tau)>C^\ast$ for
$\tau\ge\tau_0$, and then $V^\ast(\epsilon,\tau)\le C^\ast<V_0(\epsilon,\tau)$.
\qed

\subsection{Proof of Proposition~\ref{prop:finite-range}}
\label{app:proof-finite-range}

Define the root set
$\mathcal{R}_T=\{(a_{\epsilon,\tau}^\ast,\tau):0<\tau\le T\}$.
By the assumed compactness of the contaminated roots, $\mathcal{R}_T$
is contained in the compact set $K_T\times[0,T]$.
Under the assumed continuous differentiability of
$\Psi^\ast(a;\Gept)$ with respect to $(a,\tau)$ and uniqueness of the roots,
the map $\tau\mapsto a_{\epsilon,\tau}^\ast$ is continuous, so
$\mathcal{R}_T$ is compact.

The function $g(a,\tau)=\partial_a\Psi^\ast(a;\Gept)$ is continuous on
$K_T\times[0,T]$ by assumption, hence on $\mathcal{R}_T$.
Setting
$\underline{A}_T=\min_{\mathcal{R}_T}|g|$ and
$\overline{A}_T=\max_{\mathcal{R}_T}|g|$,
both are attained. The local convexity condition of
Lemma~\ref{lem:convexity} (which applies throughout $K_T$ by assumption)
implies $g\neq 0$ on $\mathcal{R}_T$, so $\underline{A}_T>0$.
Continuity on a compact set gives $\overline{A}_T<\infty$. \qed

\subsection{Proof of Theorem~\ref{thm:consistency}}
\label{app:proof-consistency}

Because $\aopt$ is the unique minimizer of $J$, for every $\varepsilon>0$
there exists $\eta>0$ with
$\inf_{|a-\aopt|\ge\varepsilon}J(a)\ge J(\aopt)+3\eta$.
On the event $\sup_a|J_n(a)-J(a)|<\eta$,
$J_n(\aopt)\le J(\aopt)+\eta$, while for $|a-\aopt|\ge\varepsilon$,
$J_n(a)\ge J(\aopt)+2\eta>J_n(\aopt)$.
Hence $\anstar$ cannot lie outside the $\varepsilon$-ball around $\aopt$,
giving $\Prob(|\anstar-\aopt|\ge\varepsilon)\to 0$.
Under correct specification, $\aopt=\mzero$, and consistency
$\anstar\pto\mzero$ follows. \qed

\subsection{Proof of Proposition~\ref{prop:avgopt}}
\label{app:proof-avgopt}

Posterior concentration of $\Pi_n$ at $\sigma_\epsilon^2$ and
continuity of $J(a;\sigma^2)$ imply $Q_n(a)\to J(a;\sigma_\epsilon^2)$
uniformly over $\mathcal{A}$; the argmin theorem gives
$a_n^*\to a_\epsilon^*$.

For the expansion, $a_\epsilon^*$ satisfies $\nabla_a
J(a_\epsilon^*;\sigma_\epsilon^2)=0$ for all $\epsilon$.
Differentiating this identity with respect to $\epsilon$ and applying
the chain rule,
\[
\partial^2_a J(a_\epsilon^*;\sigma_\epsilon^2)
\cdot\frac{da_\epsilon^*}{d\epsilon}
+
\nabla_{\sigma^2}\nabla_a J(a_\epsilon^*;\sigma_\epsilon^2)
\cdot\frac{d\sigma_\epsilon^2}{d\epsilon}
= 0.
\]
Setting $\epsilon=0$ and substituting
$d\sigma_\epsilon^2/d\epsilon|_{\epsilon=0}=B_\sigma(H)$
(Theorem~\ref{thm:pseudotrue}),
\[
H_a(\sigma_0^2)\cdot\frac{da_\epsilon^*}{d\epsilon}\bigg|_{\epsilon=0}
+
\nabla_{\sigma^2}\nabla_a J(a_0^*;\sigma_0^2)\cdot B_\sigma(H)
= 0.
\]
Since $H_a(\sigma_0^2)>0$ by Lemma~\ref{lem:convexity}, solving gives
$da_\epsilon^*/d\epsilon|_{\epsilon=0}=B_a(H)=-H_a(\sigma_0^2)^{-1}\nabla_{\sigma^2}\nabla_a J(a_0^*;\sigma_0^2)\,B_\sigma(H)$,
and then the Taylor expansion yields
$a_\epsilon^*=a_0^*+\epsilon B_a(H)+o(\epsilon)$.
\qed

\end{document}